\documentclass[11pt]{article}

\usepackage[top=1.25in, bottom=1.25in, left=1.25in, right=1.25in]{geometry}
\usepackage{microtype}
\usepackage{graphicx}
\graphicspath{{figs/}}
\usepackage{subcaption}
\usepackage{booktabs} 
\usepackage{authblk}

\usepackage[hyperfootnotes=false]{hyperref}
\usepackage{enumerate}
\usepackage[authoryear,round,semicolon]{natbib}

\usepackage{amsmath}
\usepackage{amssymb}
\usepackage{mathtools}
\mathtoolsset{showonlyrefs}
\usepackage{amsthm}
\usepackage{enumerate}
\usepackage{algorithmic}
\usepackage{algorithm}
\usepackage{xcolor}
\usepackage{float}

\theoremstyle{plain}
\newtheorem{theorem}{Theorem}[section]
\newtheorem{proposition}[theorem]{Proposition}
\newtheorem{lemma}[theorem]{Lemma}

\theoremstyle{definition}
\newtheorem{definition}[theorem]{Definition}
\newtheorem{assumption}{Assumption}
\theoremstyle{remark}

\newcommand{\ATE}{\widehat{\mathrm{ATE}}}
\newcommand{\AggS}{{\scriptstyle \mathrm{ AggS}}}
\newcommand{\ArmS}{{\scriptstyle \mathrm{ ArmS}}}
\newcommand{\AggSm}{{\scriptstyle \mathrm{ AggS-m}}}
\newcommand{\ArmSm}{{\scriptstyle \mathrm{ ArmS-m}}}
\newcommand{\ERS}{{\scriptstyle \mathrm{ ERS}}}
\newcommand{\ERSnm}{{\scriptstyle \mathrm{ ERS-nm}}}
\newcommand{\ERSplus}{{\scriptstyle \mathrm{ ERS+}}}
\newcommand{\E}{\mathbb{E}}
\renewcommand{\P}{\mathbb{P}}
\newcommand{\Var}{\mathrm{Var}}
\newcommand{\R}{\mathbb{R}}
\renewcommand{\L}{\mathbb{L}}
\DeclareMathOperator*{\argmin}{argmin}
\newcommand{\D}{\mathcal{D}}
\newcommand{\sh}{{\scriptstyle \mathrm{ short}}}

\newcommand{\norm}[1]{\left\lVert #1\right\rVert}
\newcommand{\opnorm}[1]{\left\lVert #1\right\rVert_{\mathrm{op}}}

\title{Many Hands Make Light Work: \\ Long-Term Effect Inference for the Experiment Reservoir}

\author[1]{Will Hartog\thanks{This work was done while William Hartog was a PhD student at Stanford University.}}
\author[2,3]{Lihua Lei}
\author[4]{Mandar Rahurkar}
\author[2]{Navdeep Sahni}
\author[5]{Yixin Tang}

  \affil[1]{Lyft, Inc., San Francisco, CA, USA}
  \affil[2]{Graduate School of Business, Stanford University, Stanford, CA, USA}
  \affil[3]{Department of Statistics, Stanford University, Stanford, CA, USA}
  \affil[4]{DoorDash, Inc., San Francisco, CA, USA}
  \affil[5]{Galileo Technologies, Inc., Burlingame, CA, USA}

\begin{document}
\maketitle

\begin{abstract}
  The surrogate index method is the modern standard for estimation of long-term effects in a designed experiment using only short-term outcomes. This method relies on the surrogacy and comparability assumptions which may be violated in practice leading to estimation bias. We propose a new method called Experiment Reservoir Surrogate (ERS) which leverages a reservoir of historical long-term experiments that is common in large-scale experimentation platforms. ERS avoids the surrogacy assumption by predicting the long-term average treatment effect directly, and relaxes the comparability assumption by matching to similar historical experiments, without requiring every experiment in the reservoir to be comparable to the target experiment. Moreover, ERS is memory-efficient, as it only requires aggregate data for each time period, instead of raw individual-level data. We theoretically justify the improved efficiency of ERS and demonstrate the improvement empirically on a reservoir of experiments from the online platform DoorDash.
\end{abstract}

\section{Introduction}

Estimating long-term effects from short-term data is a fundamental problem of modern experimentation. In an experimental setting in which outcome data is sequential, practitioners must contend with the natural tradeoff of reducing the length of the experiment to make decisions more quickly, and taking more time to obtain better estimates of target effects. In many cases, it is realistically impossible to observe a target long-term outcome, necessitating the use of a short-term outcome as a surrogate measurement. While the theoretical basis of formal `surrogate' measurements comes from the clinical trial literature \citep{prentice89}, with many subsequent advancements \mbox{\citep{frangakis02, lauritzen04, chen07}}, recently the problem of long-term effect estimation in online experimentation platforms has pushed the field empirically and theoretically. 

Large tech companies collect massive amounts of data and are constantly running A/B tests to inform decision-making. Decision-makers care about long-term business growth and health, at the scale of months ahead, but want to make decisions quickly, at the scale of a couple weeks, necessitating the use of short-term outcomes to inform business strategy. Many companies invest resources into estimating long-term effects with short-term data, covering a variety of strategies \citep{sadeghi22, jeunen24, saito24}.

To justify and design appropriate short-term surrogates, many experimenters utilize the modern framework of \citet{athey25}, including many tech companies \citep{tripuraneni2024choosing,zhang24,bibaut2024nonparametric,shi2024using, yang2024targeting, lal2026longtermcausalinferencenoisy}. 
\citet{athey25} introduce and formalize the surrogate index, a framework for using training data to fit a prediction function that gives an estimate of the long-term outcome given short-term measurements. The surrogate index learns a prediction function using historical observational data to estimate the long-term outcome given short-term outcomes. The long-term effect inference on a new experiment comes from computing an estimated treatment effect on the imputed surrogate index for each unit.

The validity of the surrogate index relies on two key assumptions: surrogacy and comparability. The surrogacy assumption states that the long-term outcome and treatment are independent, conditional on the explanatory surrogate index computed from short-term covariates. The comparability assumption states that the training data have the same underlying relationship between short-term and long-term as the test data on which one wishes to estimate the treatment effect. In the case of \citet{athey25} the training data is historical observational data. Both assumptions concern the invariance of the functional relationship between short-term and long-term outcomes -- surrogacy implies invariance in treated and control groups, and comparability implies  invariance in short-term and long-term experiments. They are necessary for identification of the treatment effect on the long-term outcomes in the setting of \citet{athey25}.

However, in many common experimentation settings these assumptions may not hold. For example, many experiments have a continued exposure to treatment \citep{tran24} which can affect the treated unit after the short-term window. Theoretically, the surrogacy assumption implies that the outcome and treatment effect time series have the same trend over time. In settings where this is not the case, operating under the surrogacy assumption can introduce bias. We visualize the clear surrogacy violation with empirical data in Section \ref{sec:ERS_motivation}.

The comparability assumption assumes a common relationship between short-term and long-term outcomes in data that may vary greatly. The difference in experiment timing and context may result in significantly different trends, especially for observational data. In addition, the comparability assumption may be violated in settings where outcome-level metrics can be highly variable, such as in the online experimentation setting, which we visualize with empirical data in Section \ref{sec:ERS_motivation}.

There have been some attempts in relaxing the surrogacy and comparability assumptions but all come at the cost of imposing alternative assumptions. For example, \citet{kallus24} replace the surrogacy assumption by a missing-at-random assumption which leads to point identification of long-term effects even in the absence of short-term surrogates. As a result, the surrogates are used to improve statistical efficiency. 
\citet{tran24} model the experiment as a Markov Decision Process (MDP) with a binary action. Their identification strategy is based on extrapolation across time periods, rather than experiments, by exploiting the Markov property and stationarity. \citet{bibaut2024learning} leverages multiple experiments without invoking surrogacy or comparability assumptions,  but they study  different estimands than the long-term effect of a given short-term experiment.

In contrast, we aim to estimate the experiment-level treatment effect as in \citet{athey25}, but take advantage of the availability of multiple historical experiments in many tech companies. In this paper, we introduce a novel alternate procedure to predict the long-term ATE that we call Experiment Reservoir Surrogate (ERS), which predicts the long-term ATE directly using historical experiments from a large "reservoir" as units. This shift is especially applicable for large online experimentation platforms which already have a large database of completed long-term experiments, as opposed to the setting of \citet{athey25} and \citet{chen23} which only assume access to a single observational dataset, or \citet{kallus24}, which assume a single experimental setting with missing data. 

Our method modifies the surrogate methods in three ways: (a) it operates on aggregated data instead of individual-level data; (b) it fits a horizontal regression directly on the aggregated short-term and long-term ATEs instead of the outcome values, and (c) it applies a matching procedure to select the experiments from the reservoir that are most representative of the test experiment. Restricting ourselves to only aggregated data makes the statistical problem harder, but we counteract this issue through access to many historical experiments, meaning the power comes through learning across experiments. This modification substantially improves memory efficiency over the surrogate index approach, and this restriction can be worthwhile in many cases, as experiment-level aggregated data is much less computationally costly to store and access than the user-level data required for a surrogate index approach. Working with user-level data can require hours or days of compute while an aggregate analysis is on the scale of seconds or minutes.

In Section \ref{sec:our-method} we introduce ERS, and argue why it relaxes the surrogacy and comparability assumptions. We introduce two competing methods -- Aggregate Surrogate and Armwise Surrogate, as analogs of the surrogate index method that only requires aggregated data. We justify the improvement of ERS theoretically in Section \ref{sec:theory}. We start by introducing a decomposition of the mean squared error (MSE), revealing multiple noise channels that are important for method comparison. In Section \ref{sec:bias-theory}, we justify the benefit of running regressions on ATEs instead of outcomes by comparing unmatched ERS to the Aggregate Surrogate and Armwise Surrogate methods. In Section \ref{sec:theory-nonparam}, we justify the benefit of matching by comparing ERS with the variants of Aggregate Surrogate and Armwise Surrogate with matching. In Section \ref{sec:theory-w/-control}, we analyze a variant of ERS and show that ERS has a lower MSE asymptotically under an additional assumption that is supported by an empirical observation in our dataset and likely plausible for online experimentation in the tech industry. In Section \ref{sec:empirical} we demonstrate the improvement of ERS in comparison to competing methods, through an numerical exercise on a database of historical long-term experiments from DoorDash.

\section{Our Method: Experiment Reservoir Surrogate}\label{sec:our-method}

\subsection{Setting}
Our goal is to design a method that predicts the long-term treatment effect for a new online experiment, given short-term outcome data for the new experiment, and a large reservoir of completed past experiments. Formally, we assume that the user-level data consists of outcomes $\{Y_{jt}^{(i)}\}$, at a sequence of $T$ times from a set of $n$ experiments. The indices correspond to experiment $i$, user $j$, and time $t$. The users across experiments may be different.

Each (experiment, user) pair $(i,j)$ has a treatment assignment $W_{ij}$ corresponding to whether unit $j$ receives the relevant treatment in experiment $i$. As a simplification of the potential outcome framework, we write that 
\begin{align*}    Y_{jt}^{(i)}=Y_{jt}^{(i)}(W_{ij})=m_{jt,0}^{(i)}+\tau_{jt}^{(i)}W_{ij}+\epsilon_{jt}^{(i)}(W_{ij}),
\end{align*}
for expected treatment and control outcomes $m_{jt,1}^{(i)}, m_{jt,0}^{(i)}$, expected individual treatment effect $\tau_{jt}^{(i)} = m_{jt,1}^{(i)}- m_{jt,0}^{(i)}$, and mean-zero noises $\epsilon_{jt}^{(i)}(1)$ and $\epsilon_{jt}^{(i)}(0)$ conditional on $(m_{jt,0}^{(i)}, \tau_{jt}^{(i)})$. 

Here $T$ represents a long-term time of interest, at which we would like to estimate the treatment effect. In practice, $T$ is at the scale of many weeks. In our empirical illustration in Section \ref{sec:empirical}, $T=55$ days. For a new test experiment $i=0$, we assume only access to $\{Y_{jt}^{(0)}\}_{t\leq t_0}$ for a short-term cutoff $t_0$ at which point we would like to have an estimate of the long-term treatment effect at time $T$. In practice, $t_0$ is at the scale of a few weeks. In Section \ref{sec:empirical} it is set to $t_0=20$ days. For each experiment $i$, we denote by the experiment-time-level parameters $m_{t,0}^{(i)}$ and $\tau_t^{(i)}$ as
\begin{equation}\label{eq:individual_draw}
\E\left[m_{jt,1}^{(i)}\right] = m_{t,1}^{(i)}, \,\, \E\left[m_{jt,0}^{(i)}\right] = m_{t,0}^{(i)},\,\, \E\left[\tau_{jt}^{(i)}\right]=\tau_t^{(i)}.
\end{equation}

To capture the heterogeneity in the dynamics across experiments, we adopt a superpopulation framework in which each experiment has a population-level baseline $m_{t,0}^{(i)}$ and treatment effect $\tau_t^{(i)}$, which both vary across time, as we would expect from most online experiments. We assume that these are drawn from a super population.
\begin{assumption}\label{ass:pop-dist}
There exists a distribution $F$ such that 
\begin{equation*}
(m_{1,0}^{(i)},...,m_{T,0}^{(i)},\tau_1^{(i)},...,\tau_T^{(i)}) \stackrel{i.i.d.}{\sim} F.
\end{equation*}
\end{assumption}

Given the short-term data, our goal for the new experiment $i=0$ is to obtain an estimate of $\tau_T^{(0)}$,
the population-level long-term treatment effect, which may be thought of as the empirical treatment effect when the number of units is large. Unlike the surrogate index framework \citep{athey25} we propose a method that uses only aggregated experiment-level data, in particular the daily average treatment effects (ATEs), defined as
\begin{align*}
    \ATE_{t}^{(i)}&=\bar Y_{t,1}^{(i)}-\bar Y_{t,0}^{(i)} \\ \bar Y_{T,w}^{(i)}&=\frac{\sum_jY_{jt}^{(i)}\mathbf{1}\{W_{j}^{(i)}=w\}}{\sum_j\mathbf{1}\{W_{j}^{(i)}=w\}} \ \text{ for }w\in\{0,1\}.
\end{align*}
where $\{\bar Y_{T,w}^{(i)}\}$ are aggregated treatment arm-level outcome means. Here, we implicitly assume the users are randomized on the individual level, as in common practice of A/B testing. For more complex experimental designs, we can use the inverse probability weighting (IPW) or Augmented IPW (AIPW) estimators. 



\subsection{Experiment Reservoir Surrogate}\label{sec:ers}

We propose a method that leverages a large pool of past long-term historical experiments, which we call the experiment reservoir, leading to the name Experiment Reservoir Surrogate (ERS). Under the superpopulation framework (Assumption \ref{ass:pop-dist}), the design of ERS is to learn the trend in experiment-level treatment effects $\tau_t^{(i)}$ over time, and from this trend impute a long-term treatment effect $\tau_T^{(0)}$ for a new experiment where $i=0$. According to common practice for surrogate methodology, we estimate this trend with linear regression, which can be with or without regularization. However, the nonparametric and potentially complex nature of $F$ necessitates a matching component to ensure the training data is local enough that a linear trend is representative.

ERS takes a reservoir of training data $\{\ATE_{t}^{(i)}\}_{i\in[n],t\in[T]}$, and short-term test data $\{\ATE_{t}^{(0)}\}_{t\in[t_0]}$. Importantly, ERS does not require individual-level data. We build a predictor through a horizontal regression on the training reservoir, with $Y_i=\ATE_{T}^{(i)}$ as the outcome, and the vector $X_i=(1,\ATE_{1}^{(i)},...,\ATE_{t_0}^{(i)})$ as the covariates. 
We define the procedure with all short-term times $1\leq t\leq t_0$, but the experimenter may choose to restrict the number of times, which will reduce variance in the model, as we do in our empirical demonstration in Section \ref{sec:empirical}. We detail the procedure in Algorithm \ref{alg:ers}. 

Furthermore, ERS introduces a matching component to the fitting process. For a test experiment with covariates $X^{(0)}$, ERS finds the experiments $i\in[n]$ in the reservoir with covariates $X^{(i)}$ close to $X^{(0)}$, and fits the model on this subset, drawing only the experiments from the reservoir that are relevant to the current experiment.

We define a matching procedure as a deterministic function that takes as inputs an $(n,d)$-dimensional training reservoir of experimental covariates $X$, a $d$-dimensional test experiment covariate vector $X^{(0)}$, and a target number of matched experiments $k$, and outputs a $(k,d)$-dimensional subset of $X$. In this paper, we focus on matching procedures that are based on a distance metric $\delta$.

\begin{definition}\label{def:matching-alg} We define a $\delta$-matching procedure as $M_\delta(X,X^{(0)};k)\triangleq X^*$ where $\tilde X$ has dimensions $(k,d)$ and the $j$th row $\tilde X^{(j)}$ is the row of $X$ with the $i$th smallest value of $\delta_i=\delta(X^{(i)},X^{(0)})$.
\end{definition}

Definition \ref{def:matching-alg} may be modified to replace $k$ with a bandwidth hyperparameter $h$, where the subset $\tilde X$ does not have a fixed number of rows, and we select a subset of all experiments with $\delta^{(i)}\leq h$. We use this formulation of matching in Section \ref{sec:theory-nonparam} since it is easier to analyze. In practice, when the covariate space is more complex, it is more interpretable to simply use a fixed number of matched experiments $k$.

Moreover, Definition \ref{def:matching-alg} gives a flexible framework for matching with general covariates $X$, although our main focus is on covariates which are the time series of short-term ATEs. In online experimentation, other experiment-level metadata such as type of metric, region, target population, etc. may be useful in matching more appropriately, which is why we present this more general framework which experimenters may adapt to their setting. Definition \ref{def:matching-alg} is also flexible with regards to distance metric $\delta$; we discuss potential choices of metric in Section \ref{sec:empirical}.

Given a matching procedure $M_\delta$ and a pre-selected $k$, we define the ERS method for predicting the long-term ATE of a test experiment in Algorithm \ref{alg:ers}. In Algorithm \ref{alg:ers}, if a linear regression is applied to obtain $\hat{f}$, $\hat{y}^{(0)}$ can be written as $\hat{\eta}' X^{(0)}$, where $\hat{\eta}$ is the estimated coefficient vector. In addition, the number of matched experiments $k$ is an input, which should be selected beforehand by the experimenter using the training reservoir. To do this selection, we suggest a leave-one-out cross-validation-type procedure, which we explain further in Section \ref{sec:empirical}.

\begin{algorithm}
    \caption{Experiment Reservoir Surrogate (ERS)}\label{alg:ers}
    \begin{algorithmic}
        \STATE {\bfseries Input:} vectors of long-term and short-term ATEs from the reservoir $y=(\ATE_{T}^{(1)}, \ldots, \ATE_{T}^{(n)})'$ and\\ $X^{(i)}=(1, \ATE_{1}^{(i)}, \ldots, \ATE_{t_0}^{(i)})'$ for $i\in [n]$,  
        vector of short-term ATEs from the target experiment $X^{(0)}=(1, \ATE_{1}^{(0)}, \ldots, \ATE_{t_0}^{(0)})'$, distance metric $\delta$, integer $k$.

        \vspace{0.2cm}
        \STATE {\bfseries Procedure:} 
        
        \STATE Find matches \[\tilde X=M_\delta([X^{(1)}, X^{(2)}, \ldots, X^{(n)}],X^{(0)};k).\]

        \STATE Find the corresponding long-term ATEs $\tilde y$.

        \STATE Fit a predictor $\hat{f}(x)$ on $(\tilde X, \tilde y)$.
        
        \vspace{0.2cm}
        \STATE {\bfseries Output:} $\hat y^{(0)}=\hat{f}(X^{(0)})$        
    \end{algorithmic}
    
\end{algorithm}
It is evident that ERS is easy to implement and only requires access to data aggregated at the experiment level, in fact only the ATEs, whereas many other surrogate methods may require saving individual data which can be costly to both store and analyze.

\subsection{Addressing surrogacy and comparability violation through aggregate data}\label{sec:ERS_motivation}

The surrogate index method \citep{athey25} relies on two key assumptions -- surrogacy and comparability. In the context of long-term effects, both assumptions are invariance of the relationship between long- and short-term outcomes -- the former implies invariance in the treatment and control groups and the latter implies invariance in historical and target experiments. More formally, surrogacy assumes, for any unit $j$ in this in any experiment $i$, 
\begin{equation}\label{eq:surrogacy}
Y_{jT}^{(i)}(1)\mid \{Y_{jt}^{(i)}(1): t\le t_0\}\stackrel{d}{=} Y_{jT}^{(i)}(0)\mid \{Y_{jt}^{(i)}(0): t\le t_0\},
\end{equation}
and comparability assumes
\begin{equation}\label{eq:comparability}
Y_{jT}^{(i)}(0)\mid \{Y_{jt}^{(i)}(0): t\le t_0\}\stackrel{d}{=} Y_{jT}^{(0)}(0)\mid \{Y_{jt}^{(0)}(0): t\le t_0\}.
\end{equation}
Both assumptions are sensitive to the error distributions of $\epsilon_{jt}^{(i)}(1)$ and $\epsilon_{jt}^{(i)}(0)$. Specifically, even if the following experiment-level surrogacy holds:
\begin{align*}
m_{T,1}^{(i)}\mid m_{1,1}^{(i)}, \ldots,  m_{t_0,1}^{(i)}\stackrel{d}{=} m_{T,0}^{(i)}\mid m_{1,0}^{(i)}, \ldots,   m_{t_0,0}^{(i)},
\end{align*}
the assumption \eqref{eq:surrogacy} would be generally violated if $\epsilon_{jt}^{(i)}(1)$ and $\epsilon_{jt}^{(i)}(0)$ have different distributions. Similarly, even though Assumption \ref{ass:pop-dist} implies 
\begin{align*}
m_{T,0}^{(i)}\mid m_{1,0}^{(i)}, \ldots,  m_{t_0,0}^{(i)} \stackrel{d}{=} m_{T,0}^{(0)}\mid m_{1,0}^{(0)}, \ldots,   m_{t_0,0}^{(0)},
\end{align*}
The assumption \eqref{eq:comparability} fails if $\epsilon_{jt}^{(i)}(w)$ has an experiment-dependent distribution. 

Therefore, using aggregate data not only saves compute and memory but also relaxes surrogacy and comparability assumptions that are required for the surrogate index method. This benefit, however, comes at a cost because aggregation reduces each experiment to a single observation. As a result, to make the strategy viable, one must therefore have access to at least $t_0$, or typically a much larger set of, experiments in order to run a meaningful linear regression of the long-term outcomes on the short-run outcomes. 


\subsection{Other motivations of ERS: some empirical findings}\label{subsec:empirical_findings}

We construct a dataset of daily average outcome levels for both the treatment and control groups from $196$ experiments in DoorDash. Each experiment has multiple metrics, and the index $i$ encodes an experiment-metric pair. Using this dataset, we provide empirical evidence for violation of surrogacy and comparability. In addition, we will describe two other important empirical findings that are important for the theoretical framework. Later in Section \ref{sec:theory}, we will formalize how ERS benefits from these empirical findings.

\begin{figure*}
    \centering
    \includegraphics[scale = 1]{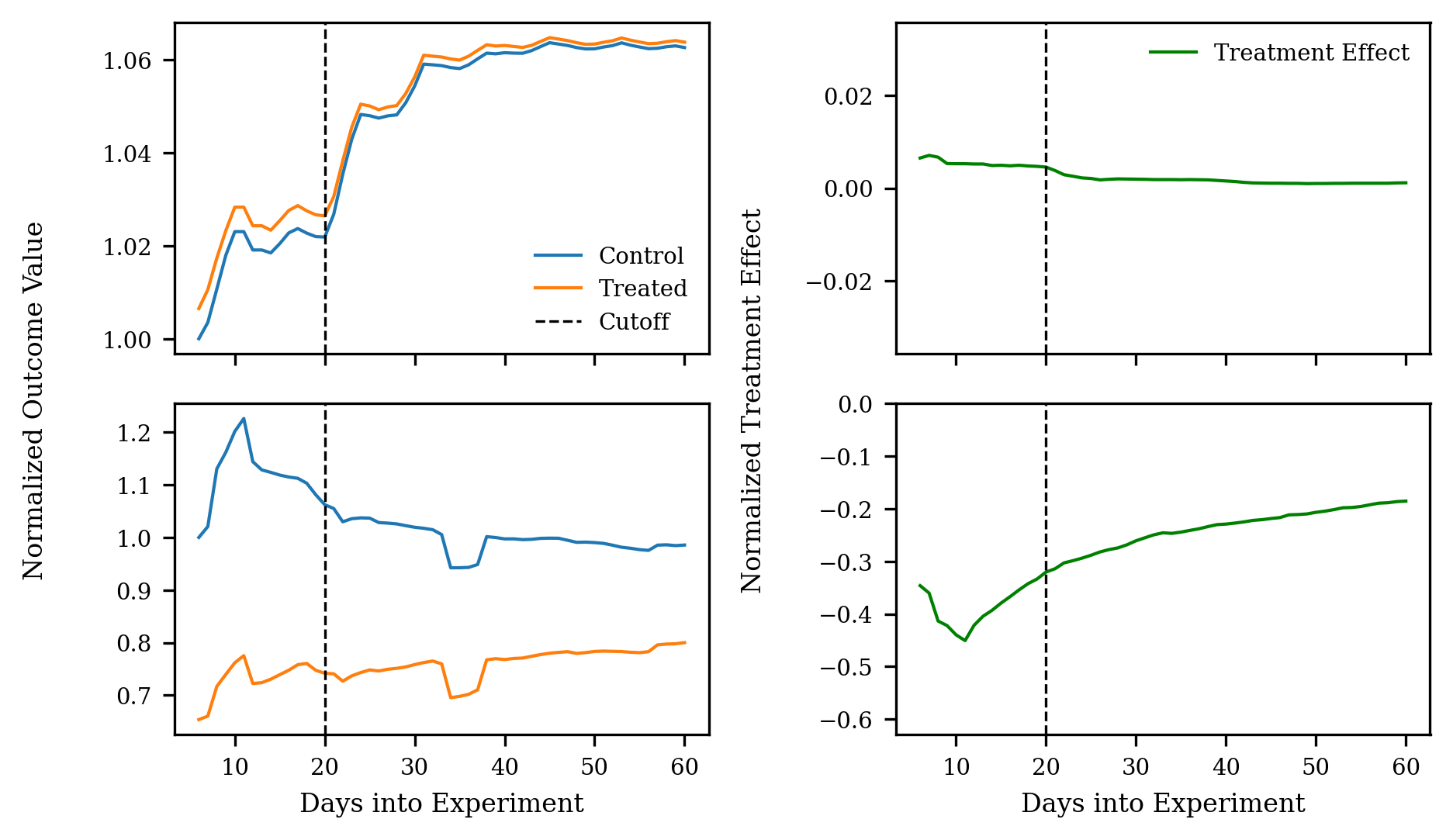}
    \caption{Two sample long-term experiments from the DoorDash database, where the top row and bottom row are the normalized outcome values and treatment effects for different experiments. Despite fluctuations in the outcome, the trend in treatment effect is more stable.}
    \label{fig:metric_treatment_plots}
\end{figure*}


\begin{figure}
    \centering 
    \includegraphics[width=0.65\linewidth]{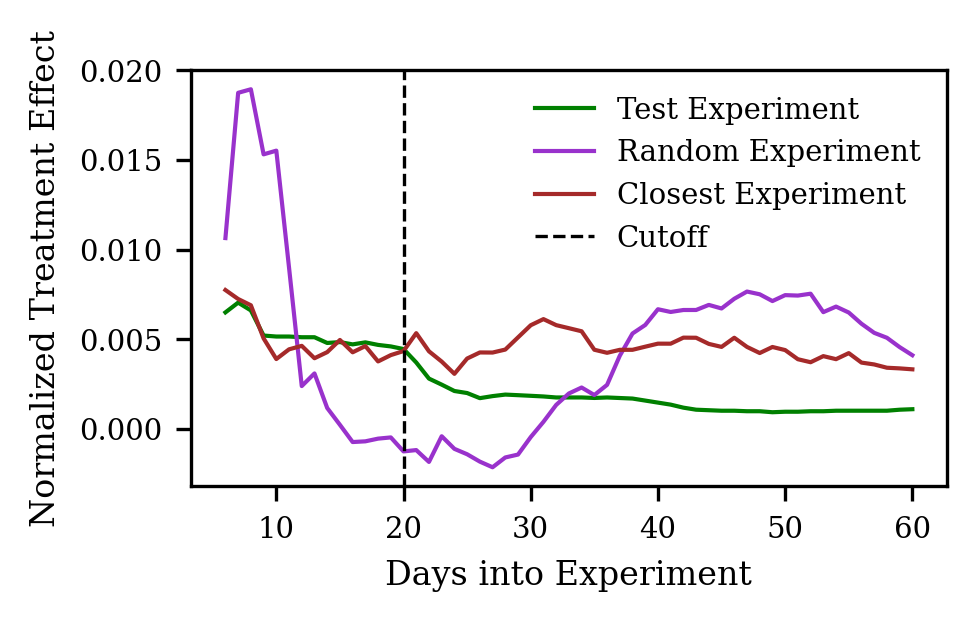}
    \caption{The ATE trend for the first experiment from Figure \ref{fig:metric_treatment_plots}, along with that of a random experiment and the closest experiment to the pre-cutoff time series in $\ell_2$ distance. The closest experiment has a much more similar trend than the random experiment.}
    \label{fig:experiment_plot}
\end{figure}

\paragraph{Empirical observation 1: ATE drives its own trend.} 
In principle, $\tau_T^{(i)}$ could be affected by outcome levels in both the treatment and control groups. Put another way, $\E[\tau_T^{(i)}\mid \tau_{1}^{(i)}, \ldots, \tau_{t_0}^{(i)}, m_{1,0}^{(i)}, \ldots, m_{t_0,0}^{(i)}]$ could depend on all variables in the conditioning set. Nevertheless, Figure \ref{fig:metric_treatment_plots} suggests that the outcome levels are barely predictive of the long-term ATE. We formalize this hypothesis in Assumption \ref{ass:cond-indep-tau-alpha}.
\begin{assumption}\label{ass:cond-indep-tau-alpha}
For any $i$,
\begin{equation}\label{eq:conditional_mean_independence}
\E[\tau_T^{(i)}\mid \tau_1^{(i)},\ldots,\tau_{t_0}^{(i)},m_{1,0}^{(i)},\ldots,m_{t_0,0}^{(i)}]
=\E[\tau_T^{(i)}\mid \tau_1^{(i)},\ldots,\tau_{t_0}^{(i)}]\qquad\text{a.s.}
\end{equation}
\end{assumption}
To assess Assumption \ref{ass:cond-indep-tau-alpha} within a model with linear conditional means, we run a linear regression of $\tau_T^{(i)}$ on $\tau_{1}^{(i)}, \ldots, \tau_{t_0}^{(i)}, m_{1,0}^{(i)}, \ldots, m_{t_0,0}^{(i)}$ and apply an F-test to the coefficients corresponding to control outcome levels, clustered by experiments. This F-test returns a p-value of $0.36$, yielding no statistically significant evidence against the restriction within this linear model. In the following, we will consider Assumption \ref{ass:cond-indep-tau-alpha} in some theoretical comparisons.

\paragraph{Empirical observation 2: ATE is less volatile. } In general, we would expect that imputing long-term values should be easier for less volatile quantities. Figure \ref{fig:metric_treatment_plots} already alludes to the idea that the treatment effect is more stable, but we also have empirical evidence that the treatment effects are less variable overall. For our empirical observation, both ATE and control mean are normalized by the control mean. For each experiment metric we compute the variance (or sum of squares) across all days in the data, including the long-term data. Then we plot a histogram of the log of this summary statistic across all experiments in Figure \ref{fig:var_ss_ratios}. We observe that most experiments have a more variable control mean, on an exponential order, indicating that on average, experiments have ATEs which are much less volatile than the control means.

\begin{figure}
    \centering
    {\includegraphics[width=0.48\textwidth]{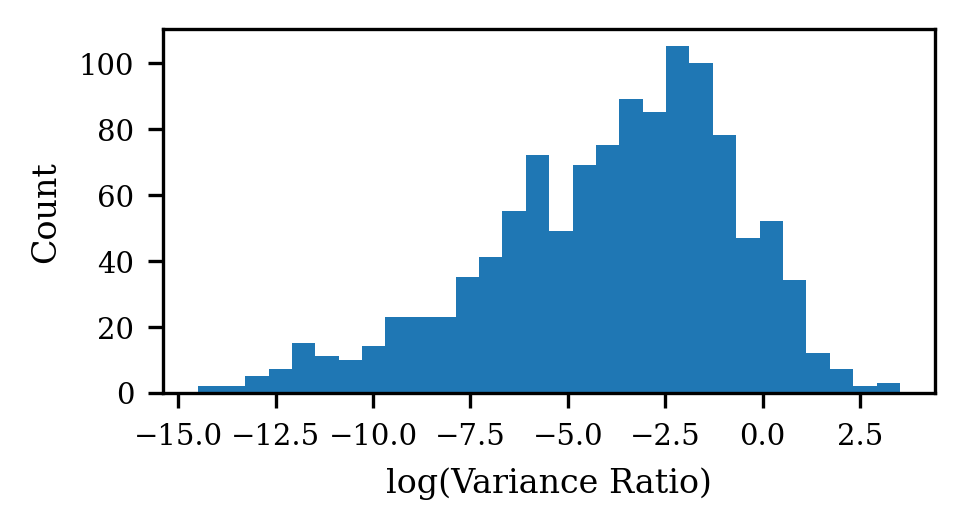}}
  \hfill
 {\includegraphics[width=0.48\textwidth]{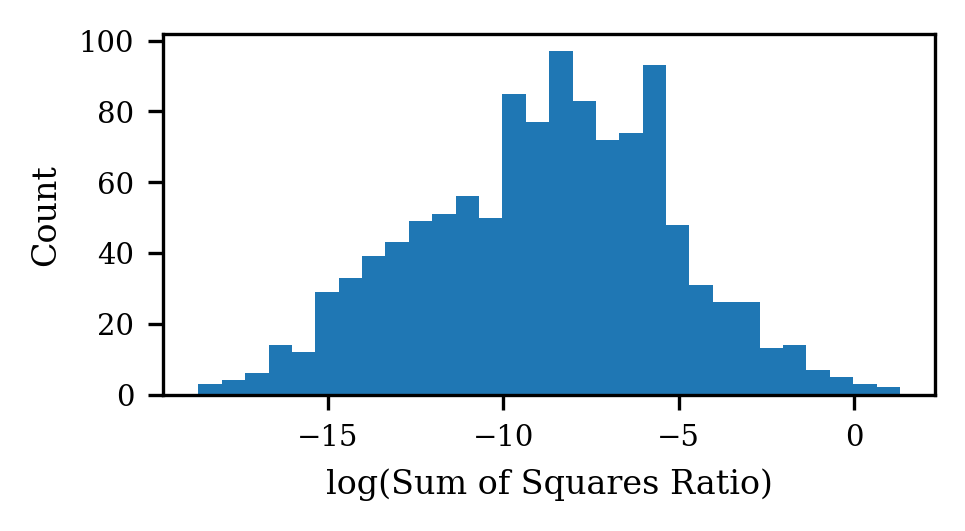}}
    \caption{For all days up until day 55, we compute the variance and sum of squares for each experiment, for both ATEs and control means, and compute the ratios. We see that ATEs in general are much less volatile than control means.}
    \label{fig:var_ss_ratios}
\end{figure}

\paragraph{Empirical observation 3: ATE trend differs from outcome trend.} Under the comparability assumption for all $T > t_0$, one would expect the trend of all experiments to be similar. Furthermore, under the surrogacy assumption, if the functional relationship is linear, we would expect $\tau$ and $m_0$ to have the same trend. We can test these assumptions jointly by computing an F-test for whether distinguishing between control mean and ATE in a regression improves the model fit. We define the outcome vector with length $2n$ as the stacked long-term control means and ATEs, corresponding stacked short-term covariates, and a dummy $D$ indicating whether the observation is of control means or ATEs. We test whether the coefficients on the interactions between $D$ and the short-term covariates are jointly zero, allowing separate intercepts for control means and ATEs. This F-test, fit with standard errors clustered by the original experiment, returns a significant p-value of $2.3\times10^{-5}$, indicating that the linear trends of ATE and control mean differ.

\paragraph{Empirical observation 4: ATE trend is more comparable for closer experiments.} Anecdotally, we show in Figure \ref{fig:experiment_plot} a test experiment as well as its closest experiment and a random experiment, visualizing that close experiments may have more similar trends. When the conditional mean of the long-term ATE is a common linear function of the short-term ATEs, adding a proximity indicator and its interactions should not improve the fit. To test this, we run an F-test for each metric within each test experiment, for a total of about 1000 metrics. We find the 50 closest experiments in Euclidean distance to the test experiment, excluding the metrics that share the same original experiment, and assign a dummy variable $G=1$ to these experiments and the test experiment and $G=0$ to the remaining experiments. Then we train a linear regression on the long-term ATE given the short-term ATEs, including the dummy $G$ and its interaction with the short-term ATEs. We run an F-test on the dummy and its interactions, which tells us whether the added information of being in a metric cluster chosen by closeness of short-term ATEs is significant. We show a histogram of these F-test p-values in Figure \ref{fig:comp-f-test-pvals}, with the left tail of about 15\% of experiments cut off. About 77\% of the p-values are under the nominal significance threshold of 0.05.

\begin{figure}
    \centering
    \includegraphics[scale = 1]{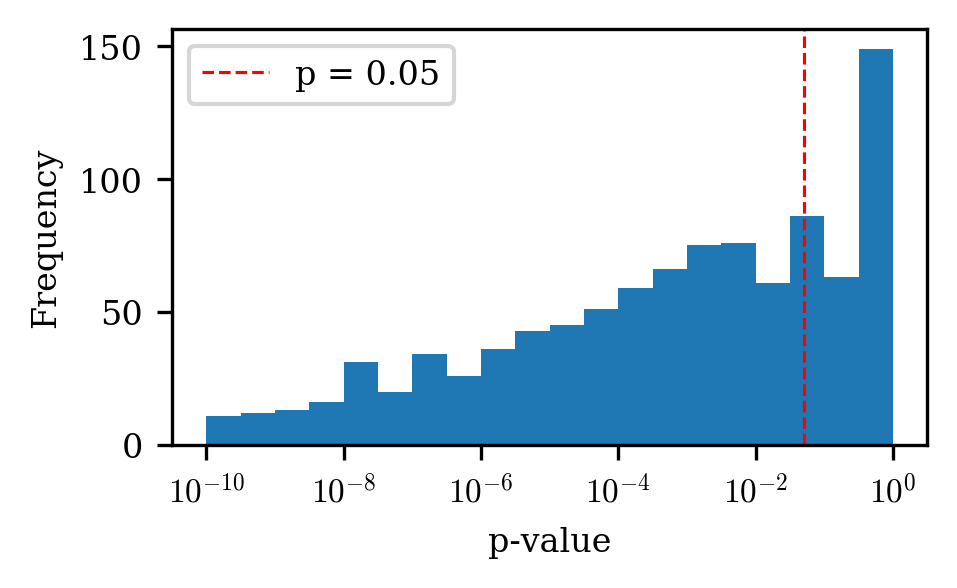}
    \caption{Histogram of the p-values for 1000 F-tests on a dummy for metric cluster and its interactions with short-term ATEs. The left tail of 15\% of metrics is cut off and about 77\% of metrics are below the nominal significance threshold of 0.05.}
    \label{fig:comp-f-test-pvals}
\end{figure}

\paragraph{Empirical observation 5: ATE and control mean are not too correlated.}

The plots in Figure \ref{fig:metric_treatment_plots} indicate that the trends of ATE and outcome level may not align, and also that they may not be too correlated. We conduct a simple analysis to show that they are not too correlated. We run a canonical correlation analysis (CCA) on the short-term ATEs $\tau_{1:t_0}$ and $m_{0,1:t_0}$, which computes the linear combination of the covariates that is maximally correlated across experiments. We compute the first CCA direction on the first 100 experiments, and compute the correlation in this direction for the remaining 96 experiments, obtaining a correlation of 0.244, which is not too large. Note that this exercise differs from the previous empirical observations: this speaks to the marginal correlation of ATE and control.

\subsection{Points of comparison: Aggregate Surrogate and Armwise Surrogate}\label{sec:agg-surr}

Throughout the rest of the paper we argue why ERS can improve on a standard outcome-based surrogate procedure. However, for the purposes of this paper, we cannot make a direct comparison between ERS and the surrogate index, since our empirical data is aggregated at the experiment level, not the user level, which the surrogate index requires. Moreover, as outlined in Section \ref{sec:ers}, using a method that only takes aggregate data as an input has many computational and memory benefits, so we want our traditional point of comparison surrogate method to have the same aggregate inputs.

We introduce Aggregate Surrogate, which functions like the surrogate index, but with experimental arms as units. This is analogous to the surrogate index method described in \citep{zhang24}. We expect that Aggregate Surrogate, on many experiments with the same outcome variable, will have an output similar to that of the surrogate index. Formally, Aggregate Surrogate uses the reservoir to fit a training function $\hat g(\bar Y_{1,w}^{(i)},...,\bar Y_{t_0,w}^{(i)})$ as an estimate of $\bar Y_{T,w}^{(i)}$, and then predicts the long-term ATE on a new test experiment with $i=0$ as $\hat\tau_T^{(0)}=\hat g(\bar Y_{11}^{(0)},...,\bar Y_{t_01}^{(0)})-\hat{g}(\bar Y_{10}^{(0)},...,\bar Y_{t_00}^{(0)})$.

For the formal algorithm for Aggregate Surrogate, we have from the training reservoir a matrix of arm-level short-term outcome means $\begin{bmatrix}X_0 \\ X_1\end{bmatrix}$ where $X_w := \{\bar Y_{T,w}^{(i)}\}_{i\in[n],t\in[t_0]}\in\R^{n\times t_0}$ and a vector of arm-level long-term outcome means $\begin{bmatrix}Y_0 \\ Y_1\end{bmatrix}$ where $y_w := \{\bar Y_{T,w}^{(i)}\}_{i\in[n]}\in\R^{n}$. For the test experiment we have two arm-level short-term outcome means $X_w^{(0)} := \{\bar Y_{t,w}^{(0)}\}_{t\in[t_0]}\in\R^{t_0}$ for $w\in\{0,1\}$. The training data $X$ and $Y$ build a linear predictor which is applied to the test data $(X_1^{(0)},X_0^{(0)})$ to get the predicted long-term ATE.

\begin{algorithm}
\caption{Aggregate Surrogate (AggS)}\label{alg:agg-surr}
    \begin{algorithmic}
        \STATE {\bfseries Input:} vectors of long-term and short-term group averages from the reservoir $Y_w=(\bar Y_{T,w}^{(1)}, \ldots, \bar Y_{T,w}^{(n)})'$ and $X_w^{(i)} = (1, \bar Y_{1,w}^{(i)},\ldots, \bar Y_{t_0,w}^{(i)})'$ for $i\in [n]$ and $w\in\{0,1\}$, vectors of short-term group averages from the target experiments $X_w^{(0)}=(1, \bar Y_{1,w}^{(0)}, \ldots, \bar Y_{t_0, w}^{(0)})'$ for $w\in\{0,1\}$.

        \vspace{0.2cm}
        \STATE {\bfseries Procedure:}

        \STATE Set $X_w = (X_w^{(1)}, \ldots, X_w^{(n)})'$ for $w\in \{0, 1\}$.
        
        \STATE Fit a predictor $\hat{g}(x)$ on $\left(\begin{bmatrix}
        X_0\\ X_1
        \end{bmatrix}, \begin{bmatrix}
            y_0 \\ y_1
        \end{bmatrix}\right)$.

        \vspace{0.2cm}
        \STATE {\bfseries Output:} $\hat{\tau}_T^{(0)} = \hat{g}(X_1^{(0)}) -\hat{g}(X_0^{(0)})$
    \end{algorithmic}
    
\end{algorithm}

When the unregularized linear regression is applied, the Aggregate Surrogate output can be factored as $\hat\eta'(X_{1}^{(0)}-X_{0}^{(0)}) = \hat\eta' X_0$, where $\hat\eta$ is the estimated coefficient vector, so it effectively decomposes into a function of the short-term ATEs. Thus, the predicted long-term effect has the same form as the ERS estimate $\hat\eta' X_0$ without the intercept, except that the coefficient vector is estimated using group averages. If the treatment effect had a different dynamics than the group averages, as illustrated in Figure \ref{fig:metric_treatment_plots}, the Aggregate Surrogate estimate would be biased. 
We formalize this intuition in Section \ref{sec:bias-theory}. 

As a result of this bias, it is natural to also consider a variant of surrogate where the long-term mean of each arm is predicted separately, and the long-term ATE is the difference of these predictions. We call this procedure Armwise Surrogate.

To disentangle the benefits and matching and ATE regression, we will also the unmatched version of ERS, ERS-nm, that runs the ATE regression on the entire reservoir, as well as the matched versions of AggS and ArmS, AggS-m and ArmS-m, that only use experiments that are close to the target as in ERS; See Appendix \ref{subapp:algorithms}, Algorithms \ref{alg:ers-unm}, \ref{alg:agg-surr-variant1} and \ref{alg:arm-surr-variant1}, for details.

\begin{algorithm}
\caption{Armwise Surrogate (ArmS)}\label{alg:arm-surr}
    \begin{algorithmic}
        \STATE {\bfseries Input:} Same as Algorithm \ref{alg:agg-surr}.

        \vspace{0.2cm}
        \STATE {\bfseries Procedure:}

        \STATE Set $X_w = (X_w^{(1)}, \ldots, X_w^{(n)})'$ for $w\in \{0, 1\}$.
        
        \STATE For each $w\in\{0,1\}$, fit a predictor $\hat{g}_w(x)$ on $\left(
        X_w, 
            y_w\right)$.

        \vspace{0.2cm}
        \STATE {\bfseries Output:} $\hat{\tau}_T^{(0)} = \hat{g}_1(X_1^{(0)}) -\hat{g}_0(X_0^{(0)})$
    \end{algorithmic}
    
\end{algorithm}

\section{Theory of ERS and other methods}\label{sec:theory}

The ERS method consists of two essential changes to Aggregate Surrogate. The first is to predict the ATE directly, instead of the treatment and control outcome values separately. The second is to match a test experiment to similar training experiments to improve the estimation. To disentangle the contribution of each component, we analyze the theoretical properties of ERS, Aggregate Surrogate (AggS), Armwise Surrogate (ArmS), and their variants. For ERS, we consider two variants: ERS without matching (ERS-nm), which applies the horizontal regression of estimated ATEs to the entire reservoir (Algorithm \ref{alg:ers-unm}), and ERS with control arm mean outcome as extra covariates (ERS+) in the horizontal regression (Algorithm \ref{alg:ers-with-control}). For AggS and ArmS, we consider a variant with matching; AggS-m in Algorithm \ref{alg:agg-surr-variant1} and ArmS-m in \ref{alg:arm-surr-variant1}).

\subsection{Asymptotic regime, target estimand, and MSE decomposition}
We are treating experiments as units in this framework, so we establish results for ERS under a distribution for these experiments. In particular, as described in Section \ref{sec:our-method}, we adopt a superpopulation framework under the assumptions \eqref{ass:pop-dist} and \eqref{eq:individual_draw}. To focus on experiment-level variations, we assume the number of individuals in the experiments is sufficiently large such that the individual-level noise is negligible. Formally, we assume 
\begin{equation}\label{eq:infinite_sample_size_experiment}
\bar{Y}_{t,w}^{(i)} = m_{t,w}^{(i)} \,\, (w=0,1), \quad \widehat{\mathrm{ATE}}_t^{(i)} = \tau_{t}^{(i)}.
\end{equation}

\noindent Our estimand is $\tau_T^{(0)}$, the long-term effect of the target experiment. 

For notational convenience, we let  \[m_{\sh, w}^{(i)} = (m_{1,w}^{(i)}, \ldots, m_{t_0, w}^{(i)}), \,\, \tau_{\sh}^{(i)} = (\tau_1^{(i)}, \ldots, \tau_{t_0}^{(i)}),\] 
and let $\D$ denote all available experiment-level data:
\[\D^{(i)} = (m_{\sh,0}^{(i)}, \tau_{\sh}^{(i)}).\]
Given an estimator $\hat{\tau}  = \hat{g}_n(\D^{(0)})$, where $\hat{g}_n$ is estimated using the reservoir
\begin{equation}\label{eq:experiment_reservoir}
\mathcal{R} := \{(\D^{(i)}, m_{T,0}^{(i)}, \tau_{T}^{(i)}): i=1,\ldots,n\},
\end{equation}
we are interested in the MSE conditional on the experiment reservoir:
\[\mathrm{MSE}(\hat{\tau}) = \E^{(0)}\left[\left\{\hat{g}_n(\D^{(0)}) - \tau_T^{(0)}\right\}^2\right],\]
where $\E^{(0)}[\cdot] = \E[\cdot \mid \mathcal{R}]$; that is, $\E^{(0)}$ treats $\hat{g}_n$ as fixed and averages over the target experiment $(\D^{(0)}, \tau_T^{(0)})$. Similarly, $\Var^{(0)}$ and $\mathrm{Cov}^{(0)}$ denote conditional variance and covariance given $\mathcal{R}$.
In particular, it measures  average performance over all potential realizations of target experiments drawn from the superpopulation $F$. Note that $\mathrm{MSE}(\hat{\tau})$ is stochastic due to the random reservoir.

Since $\tau_T^{(0)}$ is a random variable, we shall generally not expect $\mathrm{MSE}(\hat{\tau})$ to vanish, unless $\tau_{T}^{(0)}$ is a deterministic function of $\D^{(0)}$. As a consequence, we cannot define (asymptotic) unbiasedness through consistency as usual. Instead, we consider an MSE decomposition that separates irreducible errors that purely arise from randomness of $\tau_T$ from other noise channels.
\begin{proposition}\label{prop:MSE_decomp}
Consider an estimator $\hat{\tau}  = \hat{g}_n(\D^{(0)})$. Assume that, as $n\rightarrow \infty$, $\Var^{(0)}\left(\hat{g}_n(\D^{(0)}) - g(\D^{(0)})\right) = o_\P(1)$ for some deterministic function $g$.  Then
\begin{align*}
& \mathrm{MSE}(\hat{\tau})  = \underbrace{\E^{(0)}\left[\Var\left(\tau_T^{(0)}\mid \D^{(0)}\right)\right]}_{\text{irreducible error}}\\
& \quad + \underbrace{\left(\E^{(0)}\left[\hat{g}_n(\D^{(0)})\right] - \E\left[\tau_T^{(0)}\right]\right)^2}_{\text{squared bias } =: \mathrm{Bias}^2(\hat{\tau})}\\
& \quad +  \underbrace{\Var^{(0)}\left(g(\D^{(0)}) - \E\left[\tau_T^{(0)}\mid \D^{(0)}\right]\right)}_{\text{approximation error} =: \mathrm{AE}(\hat{\tau})}\\
& \quad + \underbrace{O_\P\left(\sqrt{\Var^{(0)}[\hat{g}_n(\D^{(0)}) - g(\D^{(0)})]}\right)}_{\text{reducible error} =: \mathrm{RE}(\hat{\tau})}.
\end{align*}
In addition, if $\hat{g}_n(\D^{(0)}) = \tilde{g}_n(h(\D^{(0)}))$ and $g(\D^{(0)}) = \E[\tau_T^{(0)}\mid h(\D^{(0)})]$ for some deterministic function $h$, 
\begin{equation}\label{eq:RE_information_loss}
\mathrm{RE}(\hat{\tau}) = \Var^{(0)}\left(\hat{g}_n(\D^{(0)}) - g(\D^{(0)})\right).
\end{equation}

\end{proposition}
The decomposition in Proposition \ref{prop:MSE_decomp} lays the foundation for comparing methods. In particular, the irreducible error is constant across methods. The remaining terms characterize three factors that affect the performance of a method. 
\begin{itemize}
\item The squared bias, denoted by $\mathrm{Bias}^2(\hat{\tau})$, captures whether the estimate centers at the expected long-term effect in the population of experiments. 
\item The approximation error, denoted by $\mathrm{AE}(\hat{\tau})$, depends on the choice of the function class for $\hat{g}_n$. It is driven by misspecification or information loss. For example, if $\E[\tau_T^{(0)}\mid \D^{(0)}]$ is nonlinear in $\D^{(0)}$ but $\hat{g}_n$ is estimated from a linear class (misspecification), the term inside of the variance is non-zero. Similarly, if $\hat{g}_n$ is obtained through a flexible nonparametric method but only leverages a subset of information from $\D^{(0)}$ (information loss), the approximation error is also non-zero. 
\item The reducible error, denoted by $\mathrm{RE}(\hat{\tau})$, may have either sign. It arises from the finite sample of experiments and tends to vanish as the number of experiments $n\rightarrow \infty$. The bounds below apply to its absolute value. Equation \eqref{eq:RE_information_loss} gives a sharper bound under information loss than the generic bound under misspecification.
\end{itemize}

\subsection{Benefit of ATE regression: analysis of AggS, ArmS, and ERS-nm}\label{sec:bias-theory}

In this subsection we assume $\hat{f}$ is obtained by an unregularized linear regression. For notational convenience, we write $(m_{\sh, 0}, \tau_{\sh}, m_{T,0}, \tau_{T})$ without superscripts for a generic draw from $F$. Given a random variable $y\in \R$ and a random vector $X\in \R^d$, define the linear projection of $y$ on $X$ as 
\begin{equation}\label{eq:linear_projection}
\L(y\mid X) = \mu(y\mid X) + X'\beta(y\mid X), 
\end{equation}
where 
\[\left(\mu(y\mid X), \beta(y\mid X)'\right) := \argmin_{\mu\in \R, \beta\in \R^d}\E\left[\left\{y - \mu - X'\beta\right\}^2\right].\] 

We first show that ArmS and ERS without matching are asymptotically unbiased while AggS is asymptotically biased, under the following mild moment condition on $m_{t,w}^{(i)}$. 

\begin{assumption}\label{ass:moment}
For any $i\in [n], t\in [T]$ and $w\in \{0, 1\}$,
\[\E\left[\left(m_{t,w}^{(i)}\right)^4\right] < M, \quad \text{for some constant }M<\infty.\]
\end{assumption}

In all following theorems, the $O_\P$ terms are taken with respect to the variation induced by the experiment reservoir (that is, through $\hat{g}_n$).

\begin{theorem}\label{thm:bias_unmatched}
Assume Assumptions \ref{ass:pop-dist} and \ref{ass:moment} hold and the population Gram matrices satisfy
\begin{equation}\label{eq:unmatched_gram_condition}
\E[(1\,\,X')'(1\,\,X')]\succ0
\qquad\text{for }X\in\{\tau_{\sh},m_{\sh,0},m_{\sh,1}\}.
\end{equation}
\begin{itemize}
\item \textbf{Reducible error.}
\begin{equation}\label{eq:unmatched_RE_ERSnm}
|\mathrm{RE}(\hat{\tau}_{\AggS})|,\ |\mathrm{RE}(\hat{\tau}_{\ArmS})|
=O_\P\left(\frac{1}{\sqrt n}\right),\qquad
\mathrm{RE}(\hat{\tau}_{\ERSnm})=O_\P\left(\frac1n\right).
\end{equation}
\item \textbf{Square bias.}
\[\mathrm{Bias}^2(\hat{\tau}_{\ArmS}), \mathrm{Bias}^2(\hat{\tau}_{\ERSnm}) = O_\P\left(\frac{1}{n}\right), \qquad \mathrm{Bias}^2(\hat{\tau}_{\AggS}) = O_\P(1).\]
Without further restrictions on $F$, $\mathrm{Bias}(\hat{\tau}_{\AggS})$ is generally bounded away from zero.
If the long-term outcomes have the same linear projections on the short-term outcomes across the treatment and control groups, i.e.,
\begin{equation}\label{eq:experiment_surrogacy}
\left(\mu(m_{T,1}\mid m_{\sh, 1}), \beta(m_{T,1}\mid m_{\sh, 1})'\right) = \left(\mu(m_{T,0}\mid m_{\sh, 0}), \beta(m_{T,0}\mid m_{\sh, 0})'\right),
\end{equation}
then 
\[\mathrm{Bias}^2(\hat{\tau}_{\AggS}) = O_\P\left(\frac{1}{n}\right).\]
Without further constraints on the distribution of $m_{\sh, w}$, \eqref{eq:experiment_surrogacy} is necessary for the bias of AggS to vanish, i.e., $\mathrm{Bias}^2(\hat{\tau}_{\AggS}) = o_\P(1)$.
\item \textbf{Approximation error. }
\begin{align*}
\mathrm{AE}(\hat{\tau}_{\ArmS}) &= \Var(\L(m_{T,1}\mid m_{\sh, 1}) - \L(m_{T,0}\mid m_{\sh, 0}) - \E[\tau_T\mid \tau_{\sh}, m_{\sh, 0}]).
\end{align*}
and 
\begin{align*}
\mathrm{AE}(\hat{\tau}_{\ERSnm}) &= \Var(\L(\tau_T\mid \tau_{\sh}) - \E[\tau_T\mid \tau_{\sh}, m_{\sh,0}]).
\end{align*}
Under \eqref{eq:experiment_surrogacy}, 
\begin{align*}
\mathrm{AE}(\hat{\tau}_{\ERSnm}) \le \mathrm{AE}(\hat{\tau}_{\AggS}) & = \Var(\tau_{\sh}' \beta(m_{T,0}\mid m_{\sh, 0})- \E[\tau_T\mid \tau_{\sh}, m_{\sh, 0}]).
\end{align*}
\end{itemize}
\end{theorem}
Theorem \ref{thm:bias_unmatched}  implies that (a) ERS-nm has a lower upper bound on reducible error than ArmS and AggS; (b) AggS is biased in general unless an experiment-level surrogacy assumption \eqref{eq:experiment_surrogacy} holds while the other two methods are approximately unbiased; (c) even when AggS is approximately unbiased, ERS-nm still has a smaller approximation error. All of these results suggest that the ATE regression in ERS-nm is superior to the pooled outcome regression in AggS. 

To understand the benefit of ATE regression over arm-wise outcome regression, we need to compare the approximation errors of ERS-nm and ArmS.

\begin{theorem}\label{thm:AE_comparison_ERSnm_ArmS}
Assume Assumptions \ref{ass:pop-dist} and \ref{ass:moment} hold. If
\begin{equation}\label{eq:AE_projection_condition}
\L(\tau_T\mid\tau_{\sh},m_{\sh,0})
=\L(\tau_T\mid\tau_{\sh})\qquad\text{a.s.},
\end{equation}
then
\begin{equation}\label{eq:AE_projection_comparison}
\mathrm{AE}(\hat{\tau}_{\ArmS})
\ge \mathrm{AE}(\hat{\tau}_{\ERSnm}).
\end{equation}
In particular, \eqref{eq:AE_projection_condition} holds if Assumption \ref{ass:cond-indep-tau-alpha} holds and $\E[\tau_T\mid\tau_{\sh}]$ is linear in $\tau_{\sh}$.
\end{theorem}

Theorem \ref{thm:AE_comparison_ERSnm_ArmS} gives a weak approximation-error comparison under the linear-projection condition \eqref{eq:AE_projection_condition}. Assumption \ref{ass:cond-indep-tau-alpha} alone does not imply this linear-projection condition. Small short-term ATEs and weak dependence between $\tau_{\sh}$ and $m_{\sh,0}$ alone do not guarantee the approximation-error ordering. The proof is in Appendix \ref{subapp:proof_thm_AE_comparison_ERSnm_ArmS}.

\subsection{Benefit of matching: analysis of AggS-m, ArmS-m, and ERS}\label{sec:theory-nonparam}

In this subsection we quantify the benefit of matching. In particular, we derive the MSE decomposition for ERS (Algorithm \ref{alg:ers}), AggS with matching (AggS-m), and ArmS with matching (ArmS-m). 
As mentioned below Definition \ref{def:matching-alg}, instead of sticking with the nearest neighbor version, we analyze the matching procedure that includes all experiments with $\delta^{(i)}\le h$ where $h$ is the bandwidth parameter. 

We start with the technical assumptions. 

\begin{assumption}\label{ass:local-linear-regularity}
For each pair of short- and long-term outcomes $(S,Y)$, let $\mathcal S\subset\R^q$ be its support. Assume $\mathcal S$ is compact and convex with nonempty interior. Take an open neighborhood $\mathcal O\supset\mathcal S$.
\begin{itemize}
\item The density $f$ on $\mathcal S$ and $\mu(m)=\E[Y\mid S=m]$ have extensions to $\mathcal O$ with bounded and continuous third-order derivatives with $\inf_{m\in\mathcal S}f(m)>0$.
\item $\sigma^2(m)=\Var(Y\mid S=m)$ is bounded and uniformly continuous on $\mathcal S$.
\item $\sup_{m\in\mathcal S}\E[(Y-\mu(S))^4\mid S=m]<\infty$.
\end{itemize}
\end{assumption}

\begin{assumption}\label{ass:matching-bandwidth}
For each local-linear regression of fixed dimension $q$,
\begin{equation}\label{eq:matching_bandwidth}
h=h_n\to0,
\qquad
\frac{nh^q}{\log n}\to\infty,
\qquad
nh^{q+2}\to\infty.
\end{equation}
\end{assumption}

The next theorem establishes bounds on the reducible error, squared bias, and approximation error of ERS, ArmS-m, and AggS-m.

\begin{theorem}\label{thm:decomp_matching}
Assume Assumptions \ref{ass:pop-dist}, \ref{ass:local-linear-regularity}, and \ref{ass:matching-bandwidth} hold with $q=t_0$.
\begin{itemize}
\item \textbf{Reducible error.}
\begin{align}
|\mathrm{RE}(\hat\tau_{\ERS})|
&=O_\P\left(h^4+\frac{1}{nh^{t_0}}\right),
\label{eq:matching_RE_ERS}\\
|\mathrm{RE}(\hat\tau_{\ArmSm})|,\ |\mathrm{RE}(\hat\tau_{\AggSm})|
&=O_\P\left(h^2+\frac{1}{\sqrt{nh^{t_0}}}\right).
\label{eq:matching_RE_arm_pooled}
\end{align}
\item \textbf{Square bias.} Let $b_{0,G}$ be defined in \eqref{eq:matching_pooled_bias_calculation} in Appendix \ref{subapp:proof_thm_decomp_matching}.
\begin{align}
\mathrm{Bias}^2(\hat\tau_{\ERS}),\ \mathrm{Bias}^2(\hat\tau_{\ArmSm})
&=O_\P\left(h^4+\frac{1}{n}\right),
\label{eq:matching_bias_ERS_arm}\\
\mathrm{Bias}^2(\hat\tau_{\AggSm})
&=
\begin{cases}
O_\P(h^4+n^{-1}), & b_{0,G}=0,\\
b_{0,G}^2+O_\P(h^2+n^{-1/2}), & b_{0,G}\neq0.
\end{cases}
\label{eq:matching_bias_pooled}
\end{align}
The constant $b_{0,G}$ is generally nonzero.
\item \textbf{Approximation error.} Under Assumption \ref{ass:cond-indep-tau-alpha},
\begin{equation}\label{eq:AE_matching}
\mathrm{AE}(\hat\tau_{\ERS}) = 0,
\qquad
\mathrm{AE}(\hat\tau_{\ArmSm}), \mathrm{AE}(\hat\tau_{\AggSm}) = O(1).
\end{equation}
The approximation errors of ArmS-m and AggS-m need not vanish as $h\to0$.
\end{itemize}
\end{theorem}
The proof is in Appendix \ref{subapp:proof_thm_decomp_matching}. Theorem \ref{thm:decomp_matching} gives three comparisons. (a) ERS has a smaller upper bound on reducible error than ArmS-m and AggS-m. In particular, using the optimal bandwidth for each method in \eqref{eq:matching_RE_ERS} and \eqref{eq:matching_RE_arm_pooled} gives $O_\P(n^{-4/(t_0+4)})$ for ERS and $O_\P(n^{-2/(t_0+4)})$ for ArmS-m and AggS-m. (b) ERS and ArmS-m are asymptotically unbiased, whereas AggS-m remains biased in general. (c) Under Assumption \ref{ass:cond-indep-tau-alpha}, ERS has zero approximation error, while ArmS-m and AggS-m generally have positive approximation errors. Therefore, under these assumptions, ERS offers advantages in both robustness and efficiency over ArmS-m and AggS-m.

Compared with ERS-nm in Theorem \ref{thm:bias_unmatched}, ERS replaces the global linear projection $\L(\tau_T\mid\tau_{\sh})$ with local fits that estimate $\E[\tau_T\mid\tau_{\sh}]$. Under Assumption \ref{ass:cond-indep-tau-alpha}, this eliminates approximation error from nonlinear relationship between the long- and short-term outcomes. Under Assumption \ref{ass:cond-indep-tau-alpha} and the conditions of Theorems \ref{thm:bias_unmatched} and \ref{thm:decomp_matching}, Proposition \ref{prop:MSE_decomp} implies that ERS has a strictly lower MSE asymptotically when the approximation error of ERS-nm is bounded away from zero.

\subsection{Do control outcome levels help? Analysis of ERS+}\label{sec:theory-w/-control}

One source of information that ERS leaves unused is the short-term control means. ERS+ augments the short-term ATEs with the short-term control means, using all covariates in $\D$ to predict $\tau_T$; see Algorithm \ref{alg:ers-with-control}. Write $h_+$ for its bandwidth and define
\begin{equation}\label{eq:ersplus_regression}
r_+(\D)=\E[\tau_T\mid\D],
\qquad
\hat\tau_{\ERSplus}=\hat r_+(\D^{(0)}),
\end{equation}
where $\hat r_+$ is the local-linear regression of $\tau_T$ on $\D$. Adding control means doubles the covariate dimension from $t_0$ to $2t_0$, making matching more difficult. The following result establishes the MSE decomposition for ERS+.

\begin{theorem}\label{thm:decomp_ersplus}
Assume Assumption \ref{ass:pop-dist} holds, and Assumptions \ref{ass:local-linear-regularity} and \ref{ass:matching-bandwidth} hold for the regression $(\D,\tau_T)$ with $q=2t_0$ and bandwidth $h_+$.
\begin{itemize}
\item \textbf{Reducible error.}
\begin{equation}\label{eq:ersplus_RE}
|\mathrm{RE}(\hat\tau_{\ERSplus})|
=O_\P\left(h_+^4+\frac{1}{nh_+^{2t_0}}\right).
\end{equation}
\item \textbf{Square bias.}
\begin{equation}\label{eq:ersplus_bias}
\mathrm{Bias}^2(\hat\tau_{\ERSplus})
=O_\P\left(h_+^4+\frac1n\right).
\end{equation}
\item \textbf{Approximation error.}
\begin{equation}\label{eq:ersplus_AE}
\mathrm{AE}(\hat\tau_{\ERSplus})=0.
\end{equation}
This equality does not require Assumption \ref{ass:cond-indep-tau-alpha}.
\end{itemize}
\end{theorem}
The proof is in Appendix \ref{subapp:proof_thm_decomp_ersplus}. Comparing Theorems \ref{thm:decomp_ersplus} and \ref{thm:decomp_matching}, ERS+ retains zero approximation error even when Assumption \ref{ass:cond-indep-tau-alpha} fails, whereas ERS may have positive approximation error. Nevertheless, the empirical observation that ATEs drive their own trends in Section \ref{subsec:empirical_findings} suggests that Assumption \ref{ass:cond-indep-tau-alpha} is plausible for our data. Under Assumption \ref{ass:cond-indep-tau-alpha}, ERS and ERS+ both have zero approximation error. Matching on only $t_0$ covariates instead of $2t_0$ then gives ERS a faster rate for its MSE bound. In particular, under the conditions of both theorems, using the optimal bandwidths $h\asymp n^{-1/(t_0+4)}$ and $h_+\asymp n^{-1/(2t_0+4)}$ for these bounds gives
\begin{align}
\mathrm{MSE}(\hat\tau_{\ERS})-\E[\Var(\tau_T\mid\D)]
&=O_\P\left(n^{-4/(t_0+4)}\right),\nonumber\\
\mathrm{MSE}(\hat\tau_{\ERSplus})-\E[\Var(\tau_T\mid\D)]
&=O_\P\left(n^{-4/(2t_0+4)}\right).
\label{eq:ersplus_balanced_MSE}
\end{align}
Therefore, for applications where Assumption \ref{ass:cond-indep-tau-alpha} is plausible, the rates in \eqref{eq:ersplus_balanced_MSE} favor ERS as the more efficient choice.

\subsection{Summary of Method Comparisons}\label{sec:theory-summary}

Theorem \ref{thm:bias_unmatched} and \ref{thm:decomp_matching} highlight the advantage of ERS which combines ATE regression with matching. In particular, Theorem \ref{thm:bias_unmatched} shows that ATE regression avoids the bias that pooling treatment and control outcomes (AggS) can introduce. Theorem \ref{thm:AE_comparison_ERSnm_ArmS} shows that ATE regression has no larger approximation error than separate outcome regressions (ArmS) under the linear-projection condition \eqref{eq:AE_projection_condition}, which is compatible with our empirical observations in Section \ref{subsec:empirical_findings}. On the other hand, Theorem \ref{thm:decomp_matching} implies ATE regression works better with matching than the outcome regression. Under Assumption \ref{ass:cond-indep-tau-alpha} and other stated conditions, \eqref{eq:ersplus_balanced_MSE} and \eqref{eq:MSE_step1} imply that ERS and ERS+ both achieve the minimal limiting MSE $\E[\Var(\tau_T\mid\D)]$. ERS has a lower upper bound asymptotically because it matches on fewer variables.

\section{Application on DoorDash Experiment Data}\label{sec:empirical}

The ERS method provides an improvement in practice on real databases of historical experiments. We demonstrate this improvement with DoorDash experiment data, which is a large set of 196 experiments, each of which has one to many metrics of interest, each of which represents a treatment effect time series. In effect, the metrics represent unique experiments, albeit with a high dependence between metrics within the same parent experiment.

We evaluate the accuracy of five methods: Aggregate Surrogate, Armwise Surrogate, ERS, ERS without matching, and ERS with the control arm means as extra covariates. To represent the way our method works in practice, we will evaluate these accuracies at multiple chronological time points, representing increasing reservoir sizes. In general, the effect of a larger reservoir is directionally unclear: on the one hand, it tends to yield better matches and hence improves the quality; on the other hand, it may involve more distribution shifts from older experiments, which would worsen the accuracy.


Given our DoorDash reservoir with a total of 196 experiments, we order and label each experiment chronologically in time from $i=1$ to $i=196$. For the purposes of our evaluation we denote $t_0=20$ days to be the short-term cutoff date, and $T=55$ days to be the long-term target date. Our data consists of daily average treatment effects, and our short-term data is from days 6 through 19, giving a total of $d=14$ observations for the regression. A practitioner may choose to include more or fewer days of ATEs, taking into consideration the natural bias-variance tradeoff of the number of parameters in a regression.

We devise an evaluation method in which we construct a series of training reservoirs of experiments up until a given time, with sizes $n\in[50,140]$ in intervals of 10. For each reservoir we take the next $m=50$ experiments, with indices $[n+1,n+m]$, as the set of test experiments. For each test experiment we compute the predicted long-term ATE, and compute the root mean squared error (RMSE) for each method across the 50 test experiments. This process imitates the action of enacting a surrogate method in practice with a given reservoir and a set of incoming test experiments.

For ERS in particular, we match on Euclidean distance on the normalized short-term ATE vector, which we prefer to distance on the original vector so we match on the trend and not the magnitude of the pre-cutoff treatment effect. To tune the number of matched metrics $k$ for a given reservoir size $n$, we run a leave-one-out cross-validation, using the remaining $n-1$ reservoir experiments as the local reservoir. We impose the maximum value of $k$ to be a third of the total number of reservoir metrics, to ensure matching deviates enough from the unmatched procedure. 

We show in Figure \ref{fig:percent-improvement} the percent improvement of ERS (matched, unmatched, and match with control covariates) over Armwise Surrogate without matching, where larger is better, and the Armwise Surrogate baseline performance is at 0. Clearly, the precision gain is substantial. We also show the performance of Aggregate Surrogate which is likewise worse than the ERS methods. Figure \ref{fig:ks} in Appendix \ref{subapp:figures} shows the cross-validated value of $k$ as a percentage of the total number of metrics. It hovers around 25\%, slightly less than the maximum value of a third of the total number of metrics.

\begin{figure}
    \centering
    \includegraphics[width=0.85\linewidth]{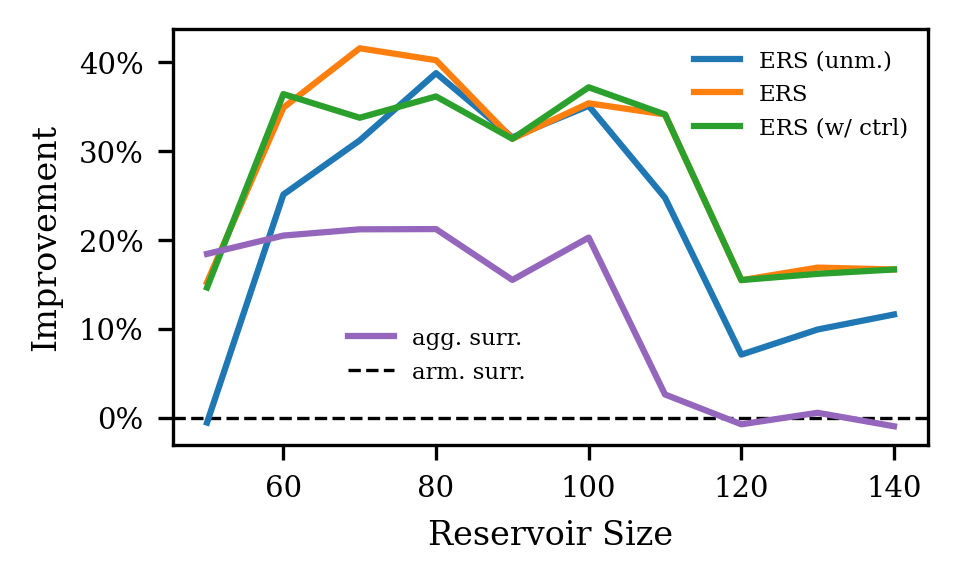}
    \caption{Percent improvement in RMSE for ERS with and without matching, compared to a baseline of ArmS. We give the comparison for a sequence of reservoir sizes, with the next 50 experiments chronologically representing the test set. The matching procedure is to cross validate on the reservoir and pick the best performing k which is less than a third of the total number of metrics in the reservoir.}
    \label{fig:percent-improvement}
\end{figure}

\section{Conclusion}

In this paper we introduce the Experiment Reservoir Surrogate (ERS) procedure, which estimates the long-term treatment effect of a new experiment given a reservoir of aggregate historical experiment data consisting of many completed long-term experiments. ERS avoids the surrogacy assumption of the surrogate index by predicting ATEs directly and relaxes the comparability assumption through matching and local smoothness, without requiring every experiment in the reservoir to be comparable to the target experiment. Moreover, ERS is efficient in both memory usage and computation as it only relies on aggregate ATE data, in contrast to the standard surrogate methods that requires individual-level data. 

We present ERS as a matching and regression on short-term treatment effects, but the framework can extend to other covariates. This flexibility is useful in the context of online experimentation, where experiment-level metadata can show which experiments may have similar trends. A further direction is to construct prediction intervals for the long-term ATE that account for both estimation error and irreducible variation. We can leverage the undersmoothing technique \citep[e.g.,][]{calonico18} or the bias-aware inference \citep[e.g.,][]{armstrong20}. Other future directions for research include selection of an optimal short-term cutoff date and dimension reduction to avoid curse of dimensionality of LLR.



\section*{Acknowledgements}

The authors would like to thank DoorDash Inc. for their collaboration and Xiaoman Luo, Dean Eckles, and Wenjing Zheng for helpful comments. The authors used ChatGPT to polish the writing and audit the proofs. The first author was supported by Achievement Rewards for College Scientists Fellowship.




\bibliography{ERS_ref}
\bibliographystyle{abbrvnat}

\newpage
\appendix
\onecolumn
\section{Proofs}

\subsection{Proof of Proposition \ref{prop:MSE_decomp}}

The reservoir $\mathcal R$ in \eqref{eq:experiment_reservoir} is independent of $(\D^{(0)},\tau_T^{(0)})$ by Assumption \ref{ass:pop-dist}. By definition,
\begin{align*}
\mathrm{MSE}(\hat{\tau}) &= \E^{(0)}\left[\left\{\hat{g}_n(\D^{(0)}) - \tau_T^{(0)}\right\}^2\right]\\
& = \E^{(0)}\left[\left\{\hat{g}_n(\D^{(0)}) - \E^{(0)}[\tau_T^{(0)}\mid \D^{(0)}] - \left(\tau_T^{(0)} - \E^{(0)}[\tau_T^{(0)}\mid \D^{(0)}]\right)\right\}^2\right]\\
& = \E^{(0)}\left[\left\{\hat{g}_n(\D^{(0)}) - \E^{(0)}[\tau_T^{(0)}\mid \D^{(0)}]\right\}^2\right] + \E^{(0)}\left[ \left\{\tau_T^{(0)} - \E^{(0)}[\tau_T^{(0)}\mid \D^{(0)}]\right\}^2\right]\\
& \qquad + 2\E^{(0)}\left[\left\{\hat{g}_n(\D^{(0)}) - \E^{(0)}[\tau_T^{(0)}\mid \D^{(0)}]\right\}\left\{\tau_T^{(0)} - \E^{(0)}[\tau_T^{(0)}\mid \D^{(0)}]\right\}\right].
\end{align*}
Since $\hat g_n$ is fixed conditional on $\mathcal R$,
\begin{align*}
& \E^{(0)}\left[\left\{\hat{g}_n(\D^{(0)}) - \E^{(0)}[\tau_T^{(0)}\mid \D^{(0)}]\right\}\left\{\tau_T^{(0)} - \E^{(0)}[\tau_T^{(0)}\mid \D^{(0)}]\right\}\mid \D^{(0)}\right]\\
& = \E^{(0)}\left[\hat{g}_n(\D^{(0)}) - \E^{(0)}[\tau_T^{(0)}\mid \D^{(0)}]\right]\cdot \E^{(0)}\left[\tau_T^{(0)} - \E^{(0)}[\tau_T^{(0)}\mid \D^{(0)}]\mid \D^{(0)}\right]\\
& = 0.
\end{align*}
Thus, 
\begin{align}
\mathrm{MSE}(\hat{\tau}) & =  \E^{(0)}\left[\left\{\hat{g}_n(\D^{(0)}) - \E^{(0)}[\tau_T^{(0)}\mid \D^{(0)}]\right\}^2\right] + \E^{(0)}\left[ \left\{\tau_T^{(0)} - \E^{(0)}[\tau_T^{(0)}\mid \D^{(0)}]\right\}^2\right]\nonumber\\
& = \E^{(0)}\left[\left\{\hat{g}_n(\D^{(0)}) - \E^{(0)}[\tau_T^{(0)}\mid \D^{(0)}]\right\}^2\right] +  \underbrace{\E^{(0)}\left[\Var^{(0)}\left(\tau_T^{(0)}\mid \D^{(0)}\right)\right]}_{\text{irreducible error}}.\label{eq:MSE_step1}
\end{align}
Next, we decompose the first term in \eqref{eq:MSE_step1}.
\begin{align}
& \E^{(0)}\left[\left\{\hat{g}_n(\D^{(0)})- \E^{(0)}\left[\tau_T^{(0)}\mid \D^{(0)}\right]\right\}^2\right]\nonumber\\
& = \E^{(0)}\left[\left\{\E^{(0)}\left[\hat{g}_n(\D^{(0)}) - \tau_T^{(0)}\mid \D^{(0)}\right]\right\}^2\right]\nonumber\\
& = \Var^{(0)}\left(\E^{(0)}\left[\hat{g}_n(\D^{(0)}) - \tau_T^{(0)}\mid \D^{(0)}\right]\right) + \left(\E^{(0)}\left[\E^{(0)}\left[\hat{g}_n(\D^{(0)}) - \tau_T^{(0)}\mid \D^{(0)}\right]\right]\right)^2\nonumber\\
& = \Var^{(0)}\left(\E^{(0)}\left[\hat{g}_n(\D^{(0)}) - \tau_T^{(0)}\mid \D^{(0)}\right]\right) + \left(\E^{(0)}\left[\hat{g}_n(\D^{(0)}) - \tau_T^{(0)}\right]\right)^2\nonumber\\
& =\Var^{(0)}\left(\hat{g}_n(\D^{(0)}) - \E^{(0)}\left[\tau_T^{(0)}\mid \D^{(0)}\right]\right) + \underbrace{\left(\E^{(0)}\left[\hat{g}_n(\D^{(0)})\right] - \E^{(0)}\left[\tau_T^{(0)}\right]\right)^2}_{\text{squared bias}}, \label{eq:MSE_step2}
\end{align}
where the second line uses the fact that $\E^{(0)}[X^2] = \Var^{(0)}(X) + \left(\E^{(0)}[X]\right)^2$. Lastly, we decompose the first term of \eqref{eq:MSE_step2}. 
\begin{align}
& \Var^{(0)}\left(\hat{g}_n(\D^{(0)}) - \E^{(0)}\left[\tau_T^{(0)}\mid \D^{(0)}\right]\right)\nonumber\\
& = \underbrace{\Var^{(0)}\left(g(\D^{(0)}) - \E^{(0)}\left[\tau_T^{(0)}\mid \D^{(0)}\right]\right)}_{\text{approximation error}} \nonumber\\
& \quad + \underbrace{\begin{aligned}
&\Var^{(0)}\left(\hat{g}_n(\D^{(0)}) - g(\D^{(0)})\right)\\
&\quad+2\mathrm{Cov}^{(0)}\left(\hat{g}_n(\D^{(0)}) - g(\D^{(0)}), g(\D^{(0)}) - \E^{(0)}\left[\tau_T^{(0)}\mid \D^{(0)}\right]\right)
\end{aligned}}_{\text{reducible error}}.\label{eq:MSE_step3}
\end{align}
By Cauchy-Schwarz inequality, 
\begin{align*}
& \bigg|\mathrm{Cov}^{(0)}\left(\hat{g}_n(\D^{(0)}) - g(\D^{(0)}), g(\D^{(0)}) - \E^{(0)}\left[\tau_T^{(0)}\mid \D^{(0)}\right]\right)\bigg|\\
& \le \sqrt{\Var^{(0)}\left(\hat{g}_n(\D^{(0)}) - g(\D^{(0)})\right)}\cdot \sqrt{\Var^{(0)}\left(g(\D^{(0)}) - \E^{(0)}\left[\tau_T^{(0)}\mid \D^{(0)}\right]\right)}.
\end{align*}
Thus, 
\[|\mathrm{RE}(\hat{\tau})| = O_\P\left(\sqrt{\Var^{(0)}\left(\hat{g}_n(\D^{(0)}) - g(\D^{(0)})\right)}\right).\]

Combining \eqref{eq:MSE_step1} - \eqref{eq:MSE_step3}, we complete the proof of the generic case. 

When $g(\D^{(0)}) = \E[\tau_T^{(0)}\mid h(\D^{(0)})]$, 
\begin{align*}
&\E\left[\left(\hat{g}_n(\D^{(0)}) - g(\D^{(0)})\right)\left(g(\D^{(0)}) - \E^{(0)}\left[\tau_T^{(0)}\mid \D^{(0)}\right]\right)\right]\\
& = \E^{(0)}\left[\left(\tilde{g}_n(h(\D^{(0)})) - \E^{(0)}[\tau_T^{(0)}\mid h(\D^{(0)})]\right)\left(\E^{(0)}[\tau_T^{(0)}\mid h(\D^{(0)})] - \E^{(0)}\left[\tau_T^{(0)}\mid \D^{(0)}\right]\right)\right]\\
& = \E^{(0)}\left[\E^{(0)}\left[\left(\tilde{g}_n(h(\D^{(0)})) - \E^{(0)}[\tau_T^{(0)}\mid h(\D^{(0)})]\right)\left(\E^{(0)}[\tau_T^{(0)}\mid h(\D^{(0)})] - \E^{(0)}\left[\tau_T^{(0)}\mid \D^{(0)}\right]\right)\right]\mid h(\D^{(0)})\right]\\
& = \E^{(0)}\left[\left(\tilde{g}_n(h(\D^{(0)})) - \E^{(0)}[\tau_T^{(0)}\mid h(\D^{(0)})]\right)\cdot\left[\E^{(0)}\left[\E^{(0)}[\tau_T^{(0)}\mid h(\D^{(0)})] - \E^{(0)}\left[\tau_T^{(0)}\mid \D^{(0)}\right]\mid h(\D^{(0)})\right]\right]\right]\\
& = 0,
\end{align*}
and 
\begin{align*}
&\E\left[g(\D^{(0)}) - \E^{(0)}\left[\tau_T^{(0)}\mid \D^{(0)}\right]\right]\\
& = \E\left[\E\left[\tau_T^{(0)}\mid h(\D^{(0)})\right] - \E^{(0)}\left[\tau_T^{(0)}\mid \D^{(0)}\right]\right]\\
& = 0.
\end{align*}
Thus, 
\[\mathrm{Cov}^{(0)}\left(\hat{g}_n(\D^{(0)}) - g(\D^{(0)}), g(\D^{(0)}) - \E^{(0)}\left[\tau_T^{(0)}\mid \D^{(0)}\right]\right) = 0.\]
By \eqref{eq:MSE_step3}, 
\[\mathrm{RE}(\hat{\tau}) = \Var^{(0)}\left(\hat{g}_n(\D^{(0)}) - g(\D^{(0)})\right).\]

\subsection{Proof of Theorem \ref{thm:bias_unmatched}}
We start by proving a lemma. 
\begin{lemma}\label{lem:linear_projection}
Let $(X_0, y_0), (X_1, y_1), \ldots, (X_n, y_n)\in \R^{d+1}$ be i.i.d. random vectors with 
\[\E[\|X_i\|^4] < \infty, \quad \E[y_i^4] < \infty, \quad \E[(1\,\, X_i')'(1\,\, X_i')]\succ0.\]
Further, let $\hat{\mu}_n$ and $\hat{\beta}_n$ be the regression coefficient estimate using all data but $(X_0, y_0)$, i.e., 
\[\begin{bmatrix}
\hat{\mu}_n(y\mid X)\\ \hat{\beta}_n(y\mid X)\end{bmatrix} = \argmin_{\mu\in \R, \beta\in \R^d}\frac{1}{n}\sum_{i=1}^{n}(y_i - \mu - X_i'\beta)^2.\]
Then 
\[\begin{bmatrix}
\hat{\mu}_n(y\mid X)\\\hat{\beta}_n(y\mid X)\end{bmatrix} = \begin{bmatrix}
\mu(y\mid X)\\
\beta(y\mid X)\end{bmatrix} + O_\P\left(\frac{1}{\sqrt{n}}\right),\]
and
\[\E^{(0)}[\hat{\mu}_n(y\mid X) + X_0'\hat{\beta}_n(y\mid X)] - \mathbb{E}^{(0)}[y_0] = O_\P\left(\frac{1}{\sqrt{n}}\right).\] 
\end{lemma}
\begin{proof} 
Let $\tilde{X}_i = (1\,\, X_i')'\in \R^{d+1}$. Let $(X, \tilde{X}, y)$ denote a generic draw. Without loss of generality, assume $\E[X_i]= 0$; otherwise, we shift all $X_i$ to $X_i - \E[X]$ and the linear prediction $\hat{\mu}_n + X_0'\hat{\beta}_n$ remains unchanged. Then 
\[\mu(y\mid X) = \E[y], \quad \beta(y\mid X) = \left(\E[XX']\right)^{-1}\E[Xy].\]
Since $X_0$ is independent of $(\hat{\mu}_n, \hat{\beta}_n)$, 
\[\E^{(0)}[\hat{\mu}_n + X_0'\hat{\beta}_n] -\E[y] = \hat{\mu}_n + \E[ X]'\hat{\beta}_n - \mathbb{E}[y] = \hat{\mu}_n - \E[y].\]
Then 
\[\E[\|\tilde{X}_i\|^4] < \infty, \quad \E[\|\tilde{X}_i y_i\|^2] < \infty.\]
By definition,
\[\begin{bmatrix}
\hat{\mu}_n \\ \hat{\beta}_n
\end{bmatrix} = \left(\frac{1}{n}\sum_{i=1}^{n}\tilde{X}_i \tilde{X}_i'\right)^{-1}\left(\frac{1}{n}\sum_{i=1}^{n}\tilde{X}_i y_i\right).\]
By Chebyshev inequality, 
\[\frac{1}{n}\sum_{i=1}^{n}\tilde{X}_i \tilde{X}_i' = \E[\tilde{X}\tilde{X}'] + O_\P\left(\frac{1}{\sqrt{n}}\right), \quad \frac{1}{n}\sum_{i=1}^{n}\tilde{X}_i y_i = \E[\tilde{X}y] + O_\P\left(\frac{1}{\sqrt{n}}\right).\]
Since $\E[\tilde{X}\tilde{X}']\succ 0$, we have
\begin{align*}
\begin{bmatrix}
\hat{\mu}_n \\ \hat{\beta}_n
\end{bmatrix} &= \left\{\E[\tilde{X}\tilde{X}'] + O_\P\left(\frac{1}{\sqrt{n}}\right)\right\}^{-1}\left\{\E[\tilde{X}y] + O_\P\left(\frac{1}{\sqrt{n}}\right)\right\}\\
& = \begin{bmatrix}
\mu(y\mid X) \\ \beta(y\mid X)
\end{bmatrix} + O_\P\left(\frac{1}{\sqrt{n}}\right).
\end{align*}
Then 
\[\E^{(0)}[\hat{\mu}_n + X_0'\hat{\beta}_n] -\E[y] = \mu(y\mid X)-\E[y] + O_\P\left(\frac{1}{\sqrt{n}}\right) = O_\P\left(\frac{1}{\sqrt{n}}\right).\]
\end{proof}

\noindent\textbf{Proof of Theorem \ref{thm:bias_unmatched}.} We prove the results for each method separately.

\bigskip
\noindent\textbf{Proof for ArmS. }Under Assumption \ref{ass:moment}, for any $t\in [t_0]$ and $w\in \{0, 1\}$,
\[\E\left[\|m_{\sh,w}^{(i)}\|^4\right] < \infty, \quad \E\left[\|m_{\sh,w}^{(i)}\|^2 \left(m_{T,w}^{(i)}\right)^2\right] < \infty\]
For ArmS defined in Algorithm \ref{alg:agg-surr-variant1}, \eqref{eq:unmatched_gram_condition} verifies the Gram-matrix condition of Lemma \ref{lem:linear_projection}, which implies, for $w\in \{0,1\}$,
\[\begin{bmatrix}
\hat{\mu}_n(m_{T,w}\mid m_{\sh, w})\\ \hat{\beta}_n(m_{T,w}\mid m_{\sh, w})
\end{bmatrix}= \begin{bmatrix}
\mu(m_{T,w}\mid m_{\sh, w})\\
\beta(m_{T,w}\mid m_{\sh, w})
\end{bmatrix}+ O_\P\left(\frac{1}{\sqrt{n}}\right),\]
and
\[\E^{(0)}\left[\hat{g}_1(X_1^{(0)})\right] = \E[m_{T,1}] + O_\P\left(\frac{1}{\sqrt{n}}\right), \quad \E^{(0)}\left[\hat{g}_0(X_0^{(0)})\right] = \E[m_{T,0}] + O_\P\left(\frac{1}{\sqrt{n}}\right).\]
Thus, 
\begin{align*}
\hat{\tau}_{\ArmS}^{(0)} & = \hat{\mu}_n(m_{T,1}\mid m_{\sh,1}) + m_{\sh, 1}^{(0)'}\hat{\beta}_n(m_{T,1}\mid m_{\sh,1})\\
&\quad - \hat{\mu}_n(m_{T,0}\mid m_{\sh,0}) - m_{\sh, 0}^{(0)'}\hat{\beta}_n(m_{T,0}\mid m_{\sh,0}),
\end{align*}
and
\[\E^{(0)}\left[\hat{\tau}^{(0)}_{\ArmS}\right] = \E[\tau_T] + O_\P\left(\frac{1}{\sqrt{n}}\right).\]
Choose $g(\D^{(0)})$ as 
\begin{align*}
g(\D^{(0)}) & = \mu(m_{T,1}\mid m_{\sh,1}) + m_{\sh, 1}^{(0)'}\beta(m_{T,1}\mid m_{\sh,1}) -  \mu(m_{T,0}\mid m_{\sh,0}) - m_{\sh, 0}^{(0)'}\beta(m_{T,0}\mid m_{\sh,0})\\
& = \L(m_{T,1}\mid m_{\sh, 1}) - \L(m_{T,0}\mid m_{\sh, 0}).
\end{align*}
As a result, 
\[\mathrm{AE}(\hat{\tau}_{\ArmS}) = \Var(\L(m_{T,1}\mid m_{\sh, 1}) - \L(m_{T,0}\mid m_{\sh, 0}) - \E[\tau_T\mid \tau_{\sh}, m_{\sh, 0}]).\]
By Lemma \ref{lem:linear_projection} and Assumption \ref{ass:moment},
\begin{align*}
& \mathrm{RE}^2(\hat{\tau}^{(0)}_{\ArmS})\\
& = O_\P\left(\sum_{w=0}^1 \left\{\hat{\mu}_n(m_{T,w}\mid m_{\sh,w}) - \mu(m_{T,w}\mid m_{\sh,w})\right\}^2 + \left\{\hat{\beta}_n(m_{T,w}\mid m_{\sh,w}) - \beta(m_{T,w}\mid m_{\sh,w})\right\}^2\right)\\
& = O_\P\left(\frac{1}{n}\right).
\end{align*}
By Lemma \ref{lem:linear_projection} again, 
\[ \mathrm{Bias}^2(\hat{\tau}^{(0)}_{\ArmS}) = O_\P\left(\frac{1}{n}\right).\]

\bigskip
\noindent\textbf{Proof for ERS-nm. } Similarly, Assumption \ref{ass:moment} implies, for any $t\in [t_0]$, 
\[\E\left[\|\tau_{\sh}^{(i)}\|^4\right] < \infty, \quad \E\left[\|\tau_{\sh}^{(i)}\|^2 \left(\tau_{T}^{(i)}\right)^2\right] < \infty.\]
For ERS-nm, \eqref{eq:unmatched_gram_condition} verifies the Gram-matrix condition of Lemma \ref{lem:linear_projection}, which implies
\[\begin{bmatrix}
\hat{\mu}_n(\tau_{T}\mid \tau_{\sh})\\ \hat{\beta}_n(\tau_{T}\mid \tau_{\sh})
\end{bmatrix}= \begin{bmatrix}
\mu(\tau_{T}\mid \tau_{\sh})\\
\beta(\tau_{T}\mid \tau_{\sh})
\end{bmatrix}+ O_\P\left(\frac{1}{\sqrt{n}}\right),\]
and
\[\E^{(0)}\left[\hat{\tau}^{(0)}_{\ERSnm}\right] = \E[\tau_T] + O_\P\left(\frac{1}{\sqrt{n}}\right).\]
Choose $g(\D^{(0)})$ as 
\[g(\D^{(0)}) = \mu(\tau_{T}\mid \tau_{\sh}) + \tau_{\sh}^{(0)'}\beta(\tau_{T}\mid \tau_{\sh}) = \L(\tau_T \mid \tau_{\sh}).\]
This implies 
\[\mathrm{AE}(\hat{\tau}_{\ERSnm}) = \Var(\L(\tau_T\mid \tau_{\sh}) - \E[\tau_T\mid \tau_{\sh}, m_{\sh,0}]).\]
For the RE bound, note that
\begingroup\medmuskip=1mu\thickmuskip=1mu
\begin{align}\label{eq:unmatched_ERSnm_estimation_error}
\hat\tau_{\ERSnm}^{(0)}-g(\D^{(0)})&=\hat\mu_n(\tau_T\mid\tau_{\sh})-\mu(\tau_T\mid\tau_{\sh})+\tau_{\sh}^{(0)\prime}\{\hat\beta_n(\tau_T\mid\tau_{\sh})-\beta(\tau_T\mid\tau_{\sh})\}.
\end{align}
\endgroup
\noeqref{eq:unmatched_ERSnm_estimation_error}
Since $g$ is the population linear projection,
$\E[(1\,\,\tau_{\sh}')'\{g(\D)-\tau_T\}]=0$.
Since $\hat\tau_{\ERSnm}^{(0)}-g(\D^{(0)})$ is linear in $\tau_{\sh}^{(0)}$ conditional on the reservoir, we obtain that
\begin{align}
&\mathrm{Cov}^{(0)}\left(\hat\tau_{\ERSnm}^{(0)}-g(\D^{(0)}),g(\D^{(0)})-
\E[\tau_T^{(0)}\mid\D^{(0)}]\right)\nonumber\\
&\qquad=\mathrm{Cov}^{(0)}\left(\hat\tau_{\ERSnm}^{(0)}-g(\D^{(0)}),g(\D^{(0)})-\tau_T^{(0)}\right)=0.
\label{eq:unmatched_ERSnm_covariance}
\end{align}
By Proposition \ref{prop:MSE_decomp}, \eqref{eq:MSE_step3}, and \eqref{eq:unmatched_ERSnm_covariance}, the RE can be written as
\begin{align}
\mathrm{RE}(\hat\tau_{\ERSnm}^{(0)})
&=\Var^{(0)}\left(\hat\tau_{\ERSnm}^{(0)}-g(\D^{(0)})\right).
\label{eq:unmatched_ERSnm_RE_decomposition}
\end{align}
By Equation \eqref{eq:unmatched_ERSnm_RE_decomposition} and Lemma \ref{lem:linear_projection},
\begin{align}
\mathrm{RE}(\hat\tau_{\ERSnm}^{(0)})
&=\Var^{(0)}\left(\hat\tau_{\ERSnm}^{(0)}-g(\D^{(0)})\right)\nonumber\\
&=\Var^{(0)}\left(\tau_{\sh}^{(0)\prime}
\{\hat\beta_n(\tau_T\mid\tau_{\sh})-\beta(\tau_T\mid\tau_{\sh})\}\right)\nonumber\\
&\le\E[\norm{\tau_{\sh}}^2]\,
\norm{\hat\beta_n(\tau_T\mid\tau_{\sh})-\beta(\tau_T\mid\tau_{\sh})}^2
=O_\P\left(\frac1n\right).
\label{eq:unmatched_ERSnm_RE_proof}
\end{align}
Thus, \eqref{eq:unmatched_RE_ERSnm} is proved.
By Lemma \ref{lem:linear_projection},
\[\mathrm{Bias}^2(\hat{\tau}^{(0)}_{\ERSnm}) = O_\P\left(\frac{1}{n}\right).\]

\bigskip
\noindent\textbf{Proof for AggS.} Let $(\hat{\mu}_{\AggS}, \hat
{\beta}_{\AggS})$ denote the estimated coefficients in the pooled regression, i.e., 
\[\begin{bmatrix}
\hat{\mu}_{\AggS} \\ \hat{\beta}_{\AggS}
\end{bmatrix} = \left(\frac{1}{n}\sum_{i=1}^{n} \left\{\tilde{X}_{i,1}\tilde{X}_{i,1}' + \tilde{X}_{i,0}\tilde{X}_{i,0}'\right\}\right)^{-1}\left(\frac{1}{n}\sum_{i=1}^{n}\left\{\tilde{X}_{i,1}m_{T,1}^{(i)} + \tilde{X}_{i,0}m_{T,0}^{(i)}\right\}\right),\]
where 
\[\tilde{X}_{i,w} = (1, m_{1,w}^{(i)}, \ldots, m_{t_0,w}^{(i)})'.\]
For notational convenience, for $w\in \{0,1\}$, let
\[\tilde{\Sigma}_w = \E[\tilde{X}_w\tilde{X}_w'], \quad A_w = \left(\tilde{\Sigma}_1 + \tilde{\Sigma}_0 \right)^{-1}\tilde{\Sigma}_w, \quad \begin{bmatrix}
\mu_w \\ \beta_w
\end{bmatrix} = \begin{bmatrix}
\mu(m_{T,w}\mid m_{\sh,w}) \\ \beta(m_{T,w}\mid m_{\sh,w})
\end{bmatrix}\]
Condition \eqref{eq:unmatched_gram_condition} gives $\tilde{\Sigma}_1\succ0$ and $\tilde{\Sigma}_0\succ0$, hence $\tilde{\Sigma}_1+\tilde{\Sigma}_0\succ0$. The coefficient expansion in the proof of Lemma \ref{lem:linear_projection}, applied to the independent experiment-level sums in the pooled fit, gives
\begin{align*}
\begin{bmatrix}
\hat{\mu}_{\AggS} \\ \hat{\beta}_{\AggS}
\end{bmatrix} &= \left(\E[\tilde{X}_1\tilde{X}_1'] + \E[\tilde{X}_0\tilde{X}_0']\right)^{-1}\left(\E[\tilde{X}_1m_{T,1}] + \E[\tilde{X}_0m_{T,0}]\right) + O_\P\left(\frac{1}{\sqrt{n}}\right)\nonumber\\
& = \left(\tilde{\Sigma}_1 + \tilde{\Sigma}_0 \right)^{-1}\left(\tilde{\Sigma}_1 \begin{bmatrix}
\mu_1 \\ \beta_1
\end{bmatrix} + \tilde{\Sigma}_0\begin{bmatrix}
\mu_0 \\ \beta_0
\end{bmatrix}\right) + O_\P\left(\frac{1}{\sqrt{n}}\right).
\end{align*}
Then $A_1 + A_0 = I_{t_0+1}$ and 
\[\begin{bmatrix}
\hat{\mu}_{\AggS} \\ \hat{\beta}_{\AggS}
\end{bmatrix} = A_1 \begin{bmatrix}
\mu_1 \\ \beta_1
\end{bmatrix} + A_0\begin{bmatrix}
\mu_0 \\ \beta_0
\end{bmatrix} + O_\P\left(\frac{1}{\sqrt{n}}\right).\]
Choose $g(\D^{(0)})$ as 
\[g(\D^{(0)}) = (0 \,\, \tau_{\sh}^{(0)})'\left\{A_1 \begin{bmatrix}
\mu_1 \\ \beta_1
\end{bmatrix} + A_0\begin{bmatrix}
\mu_0 \\ \beta_0
\end{bmatrix}\right\}.\]
By Assumption \ref{ass:moment}, we conclude that 
\[|\mathrm{RE}(\hat{\tau}_{\AggS})| = O_\P\left(\frac{1}{\sqrt{n}}\right).\]

Note that
\[\E[\tilde{X}_1] - \E[\tilde{X}_0] = \begin{bmatrix}
0 \\ 
\E\left[X_1\right] - \E\left[X_0\right]
\end{bmatrix} = \begin{bmatrix}
0 \\ 
\E\left[\tau_{\sh}^{(0)}\right]
\end{bmatrix}.\]
By definition, 
\begin{align*}
\E[\hat{\tau}_{\AggS}^{(0)}] & = \E[\tau_{\sh}^{(0)}]'\hat{\beta}_{\AggS} = \delta'\begin{bmatrix}
\hat{\mu}_{\AggS} \\ \hat{\beta}_{\AggS}
\end{bmatrix} \\
& = (\E[\tilde{X}_1] - \E[\tilde{X}_0])'A_1 \begin{bmatrix}
\mu_1 \\ \beta_1
\end{bmatrix} + (\E[\tilde{X}_1] - \E[\tilde{X}_0])'A_0\begin{bmatrix}
\mu_0 \\ \beta_0
\end{bmatrix} + O_\P\left(\frac{1}{\sqrt{n}}\right),
\end{align*}
and 
\[\E[\tau_T^{(0)}] = \E\left[m_{T,1}^{(0)}\right] - \E\left[m_{T,0}^{(0)}\right] = \E[\tilde{X}_1]\begin{bmatrix}
\mu_1 \\ \beta_1
\end{bmatrix} - \E[\tilde{X}_0]\begin{bmatrix}
\mu_0 \\ \beta_0
\end{bmatrix}.\]
Since $(\E[\tilde{X}_w], A_w, \mu_w, \beta_w)_{w=0}^1$ are free parameters with a single constraint $A_1 + A_0 = I_{t_0 + 1}$, the leading term of $\E[\hat{\tau}_{\AggS}^{(0)}]$ is not equal to $\E[\tau_T^{(0)}]$ in general. Thus, 
\[\mathrm{Bias}^2(\hat{\tau}^{(0)}_{\AggS})\text{ is bounded away from zero}.\]

Under the additional condition \eqref{eq:experiment_surrogacy}, 
\[\begin{bmatrix}
\mu_1 \\ \beta_1
\end{bmatrix} = \begin{bmatrix}
\mu_0 \\ \beta_0
\end{bmatrix}.\]
Thus, 
\[\E[\hat{\tau}_{\AggS}^{(0)}]  = (\E[\tilde{X}_1] - \E[\tilde{X}_0])'(A_1 + A_0) \begin{bmatrix}
\mu_0 \\ \beta_0
\end{bmatrix} + O_\P\left(\frac{1}{\sqrt{n}}\right) = \E[\tau_T^{(0)}] + O_\P\left(\frac{1}{\sqrt{n}}\right).\]
This implies 
\[\mathrm{Bias}^2(\hat{\tau}^{(0)}_{\AggS}) = O_\P\left(\frac{1}{n}\right).\]
To prove necessity, we note that the bias vanishes iff 
\[(\E[\tilde{X}_1] - \E[\tilde{X}_0])'A_1 \begin{bmatrix}
\mu_1 \\ \beta_1
\end{bmatrix} + (\E[\tilde{X}_1] - \E[\tilde{X}_0])'A_0\begin{bmatrix}
\mu_0 \\ \beta_0
\end{bmatrix} = \E[\tilde{X}_1]'\begin{bmatrix}
\mu_1 \\ \beta_1
\end{bmatrix} - \E[\tilde{X}_0]'\begin{bmatrix}
\mu_0 \\ \beta_0
\end{bmatrix},\]
for any $\E[X_w]$ and $\E[X_w X_w']$. Since these are free parameters, $A_1$ can take values in an open set $\mathcal{A}$.  Since the RHS does not depend on $A_1$, it must be that 
\[ (\E[\tilde{X}_1] - \E[\tilde{X}_0])'\begin{bmatrix}
\mu_1 \\ \beta_1
\end{bmatrix} = (\E[\tilde{X}_1] - \E[\tilde{X}_0])'\begin{bmatrix}
\mu_0 \\ \beta_0
\end{bmatrix}.\]
Since  $\E[\tilde{X}_1], \E[\tilde{X}_0]$ can be arbitrary, we conclude that 
\[\begin{bmatrix}
\mu_1 \\ \beta_1
\end{bmatrix} = \begin{bmatrix}
\mu_0 \\ \beta_0
\end{bmatrix}.\]

To establish the approximation error, we note that 
\[g(\D^{(0)}) = (0 \,\, \tau_{\sh}^{(0)})'\left\{A_1 \begin{bmatrix}
\mu_1 \\ \beta_1
\end{bmatrix} + A_0\begin{bmatrix}
\mu_0 \\ \beta_0
\end{bmatrix}\right\} = \tau_{\sh}^{(0)'} \beta_0.\]
Then 
\begin{align*}
\mathrm{AE}(\hat{\tau}_{\AggS}) = \Var(\tau_{\sh}' \beta_0 - \E[\tau_T\mid \tau_{\sh}, m_{\sh, 0}]).
\end{align*}
Note that the first terms of both $\mathrm{AE}(\hat{\tau}_{\AggS})$ and $\mathrm{AE}(\hat{\tau}_{\ERSnm})$ are linear functions in $(1, \tau_{\sh})$. By definition, 
\begin{align}
\L(\tau_T\mid \tau_{\sh})&= \argmin_{g \text{ linear}} \E\left[\{\tau_T - g(\tau_{\sh})\}^2\right] \nonumber\\
& = \argmin_{g \text{ linear}} \E\left[\{\E[\tau_T \mid \tau_{\sh}]- g(\tau_{\sh})\}^2\right] + \E[\Var(\tau_T \mid \tau_{\sh})]\nonumber\\
& = \argmin_{g \text{ linear}} \E\left[\{\E[\tau_T \mid \tau_{\sh}]- g(\tau_{\sh})\}^2\right]\label{eq:linear_projection_optimality}
\end{align}
Thus, 
\[\E\left[\{\L(\tau_T \mid \tau_{\sh}) - \E[\tau_T \mid \tau_{\sh}]\}^2\right]\le \E\left[\tau_{\sh}'\beta_0 - \E[\tau_T \mid \tau_{\sh}]\}^2\right].\]
Using the bias results for ERS-nm and AggS, we have 
\[\E\left[\L(\tau_T \mid \tau_{\sh}) - \E[\tau_T \mid \tau_{\sh}]\right] = 0 = \E\left[\tau_{\sh}'\beta_0 - \E[\tau_T \mid \tau_{\sh}]\right].\]
Therefore, 
\[\Var\left(\L(\tau_T \mid \tau_{\sh}) - \E[\tau_T \mid \tau_{\sh}]\right)\le \Var\left(\tau_{\sh}'\beta_0 - \E[\tau_T \mid \tau_{\sh}]\right).\]

\subsection{Proof of Theorem \ref{thm:AE_comparison_ERSnm_ArmS}} \label{subapp:proof_thm_AE_comparison_ERSnm_ArmS}

\begin{proof}
By Assumption \ref{ass:moment}, the population projections in \eqref{eq:linear_projection} exist. Define
\begin{align}
g_{\ERSnm}(\D)
&=\L(\tau_T\mid\tau_{\sh}),\nonumber\\
g_{\ArmS}(\D)
&=\L(m_{T,1}\mid m_{\sh,1})-\L(m_{T,0}\mid m_{\sh,0}).
\label{eq:AE_population_predictors}
\end{align}
Since both projections include an intercept, we have
\begin{align}
\E[g_{\ERSnm}(\D)]&=\E[\tau_T],\nonumber\\
\E[g_{\ArmS}(\D)]
&=\E[m_{T,1}]-\E[m_{T,0}]=\E[\tau_T].
\label{eq:AE_population_means}
\end{align}
Note that $\E[\E[\tau_T\mid\D]]=\E[\tau_T]$. Thus,
\begin{align}
\mathrm{AE}(\hat{\tau}_M)
&=\Var\{g_M(\D)-\E[\tau_T\mid\D]\}\nonumber\\
&=\E[\{g_M(\D)-\E[\tau_T\mid\D]\}^2],
\quad M\in\{\ERSnm,\ArmS\}.
\label{eq:AE_mean_square_errors}
\end{align}

Since $m_{\sh,1}=m_{\sh,0}+\tau_{\sh}$, \eqref{eq:AE_population_predictors} makes $g_{\ArmS}(\D)-g_{\ERSnm}(\D)$ a linear function of $\D$. Under the condition \eqref{eq:AE_projection_condition}, $g_{\ERSnm}(\D)=\L[\tau_T\mid\D]$. By the law of iterated expectations,
\begin{align}
&\E[\{\E[\tau_T\mid\D]-g_{\ERSnm}(\D)\}
       \{g_{\ArmS}(\D)-g_{\ERSnm}(\D)\}]\nonumber\\
&\qquad=\E[\{\tau_T-g_{\ERSnm}(\D)\}
       \{g_{\ArmS}(\D)-g_{\ERSnm}(\D)\}]
=0.
\label{eq:AE_projection_cross_term}
\end{align}
As a result,
\begin{align}
\mathrm{AE}(\hat{\tau}_{\ArmS})-\mathrm{AE}(\hat{\tau}_{\ERSnm})
&=\E[\{g_{\ArmS}(\D)-g_{\ERSnm}(\D)\}^2]\nonumber\\
&\quad+2\E[\{g_{\ArmS}(\D)-g_{\ERSnm}(\D)\}
              \{g_{\ERSnm}(\D)-\E[\tau_T\mid\D]\}]\nonumber\\
&=\E[\{g_{\ArmS}(\D)-g_{\ERSnm}(\D)\}^2]\ge0,
\label{eq:AE_projection_gap}
\end{align}
which proves \eqref{eq:AE_projection_comparison}.

To show the sufficiency, suppose $\E[\tau_T\mid\tau_{\sh}]$ is linear. Assumption \ref{ass:cond-indep-tau-alpha} gives
\begin{equation}
\E[\tau_T\mid\tau_{\sh}]=\E[\tau_T\mid\D].
\label{eq:AE_linear_conditional_means}
\end{equation}
Consequently,
\begin{align}
\L(\tau_T\mid\tau_{\sh},m_{\sh,0})
&=\E[\tau_T\mid\D]\nonumber\\
&=\E[\tau_T\mid\tau_{\sh}]
=\L(\tau_T\mid\tau_{\sh}),
\label{eq:AE_projection_sufficient}
\end{align}
so \eqref{eq:AE_projection_condition} holds.
\end{proof}

\subsection{Proof of Theorem \ref{thm:decomp_matching}}\label{subapp:proof_thm_decomp_matching}

We start with a lemma that bounds the integrated MSE for local-linear estimators in a slightly more general setting than the standard one \citep[e.g.][]{ruppert94}.

\begin{lemma}\label{lem:matching_local_linear_rates}
Let $\{(S_{i,j},Y_{i,j}):1\le j\le J\}_{i=1}^n$ be i.i.d. Suppose Assumptions \ref{ass:local-linear-regularity} and \ref{ass:matching-bandwidth} hold for $(S_{i,j},Y_{i,j})$ for every $i$ and $j$, with $S_{i,j}\in\R^q$, density $f$, and $\mu(s)=\E[Y_{i,j}\mid S_{i,j}=s]$. Let $\nu$ be a probability measure with bounded density supported on $\mathcal S$. Define $\hat\mu_h(s)$ as the intercept in
\begin{equation}\label{eq:matching_local_fit}
\min_{a\in\R,\ b\in\R^q}
\sum_{i=1}^n\sum_{j=1}^J
\mathbf 1\{\norm{S_{i,j}-s}\le h\}
\{Y_{i,j}-a-b'(S_{i,j}-s)\}^2.
\end{equation}
Then
\begin{align}
\int_{\mathcal S}\{\hat\mu_h(s)-\mu(s)\}^2\,d\nu(s)
&=O_\P\left(h^4+\frac{1}{nh^q}\right),
\label{eq:matching_local_L2}\\
\int_{\mathcal S}\{\hat\mu_h(s)-\mu(s)\}\,d\nu(s)
&=O_\P\left(h^2+\frac{1}{\sqrt n}\right).
\label{eq:matching_local_mean}
\end{align}
\end{lemma}

\begin{proof}
Write $a\lesssim b$ when $a\le Cb$ for a constant independent of $n,h,s$. For the local-linear fit in \eqref{eq:matching_local_fit}, define
\begin{equation}\label{eq:matching_local_basis}
K(u)=\mathbf 1\{\norm{u}\le1\},
\qquad z(u)=(1,u')',
\qquad u_{i,j}(s)=\frac{S_{i,j}-s}{h},
\qquad e_1=(1,0,\ldots,0)'.
\end{equation}
The empirical design matrix and its expectation are
\begin{align}
Q_{n,h}(s)
&=\frac{1}{nJh^q}\sum_{i=1}^n\sum_{j=1}^J
K(u_{i,j}(s))z(u_{i,j}(s))z(u_{i,j}(s))',
\nonumber\\
\E[Q_{n,h}(s)]
&=\int_{\{u:\norm{u}\le1,\ s+hu\in\mathcal S\}}
z(u)z(u)'f(s+hu)\,du.
\label{eq:matching_local_design}
\end{align}
Since $f$ is bounded on the compact set $\mathcal S$ and $\norm{z(u)}$ is bounded on $B(0,1)$,
\begin{equation}\label{eq:matching_population_design_upper_bound}
\sup_{s\in\mathcal S}\opnorm{\E[Q_{n,h}(s)]}\lesssim1.
\end{equation}

First, we derive a uniform lower bound on $\E[Q_{n,h}(s)]$. Choose a ball $B(c,r)\subset\mathcal S$, and set $L=2\{\operatorname{diam}(\mathcal S)+r\}$. For $0<h\le L$,
\begin{equation}\label{eq:matching_support_ball}
s+hB\left(\frac{c-s}{L},\frac rL\right)
=\left(1-\frac hL\right)s+\frac hL B(c,r)
\subset\mathcal S,
\qquad
B\left(\frac{c-s}{L},\frac rL\right)\subset B(0,1).
\end{equation}
Let $\rho=r/L$ and $a_s=(c-s)/L$, so $\norm{a_s}\le1/2$. For $v=(v_0,v_*)\in\R^{q+1}$,
\begin{align}
v'\E[Q_{n,h}(s)]v
&\ge \left(\inf_{\mathcal S}f\right)
\int_{B(a_s,\rho)}(v_0+v_*'u)^2\,du\nonumber\\
&=\left(\inf_{\mathcal S}f\right)|B(0,\rho)|
\left\{(v_0+v_*'a_s)^2+\frac{\rho^2}{q+2}\norm{v_*}^2\right\}.
\label{eq:matching_design_lower_bound}
\end{align}
Since $\norm{a_s}\le1/2$ and $a^2\le2(a+b)^2+2b^2$, we have
\begin{align}
v_0^2
&\le2(v_0+v_*'a_s)^2+2(v_*'a_s)^2\nonumber\\
&\le2(v_0+v_*'a_s)^2+2\norm{a_s}^2\norm{v_*}^2\nonumber\\
&\le2(v_0+v_*'a_s)^2+\norm{v_*}^2/2.
\label{eq:matching_intercept_bound}
\end{align}
As a result, \eqref{eq:matching_design_lower_bound} and \eqref{eq:matching_intercept_bound} imply that, for a constant $c_Q>0$,
\begin{equation}\label{eq:matching_population_design_bounds}
\inf_{s\in\mathcal S}\lambda_{\min}(\E[Q_{n,h}(s)])\ge c_Q.
\end{equation}

For fixed $s\in\mathcal S$ and $k,l\in\{1,\ldots,q+1\}$, define
\begin{equation}\label{eq:matching_design_cluster_entry}
V_i=\frac1J\sum_{j=1}^J K(u_{i,j}(s))
z_k(u_{i,j}(s))z_l(u_{i,j}(s)).
\end{equation}
Since $K(u)|z_k(u)z_l(u)|\le K(u)$,
\[
|V_i|\le\frac1J\sum_{j=1}^J K(u_{i,j}(s))\le1.
\]
\eqref{eq:matching_design_cluster_entry} and Jensen's inequality give
\begin{align}
\Var(V_i)
&\le\E[V_i^2]\le\frac1J\sum_{j=1}^J
\E[K(u_{i,j}(s))^2z_k(u_{i,j}(s))^2z_l(u_{i,j}(s))^2]\nonumber\\
&\le\frac1J\sum_{j=1}^J\P(\norm{S_{i,j}-s}\le h)
=\int_{\mathcal S\cap B(s,h)}f(x)\,dx\nonumber\\
&\le h^q|B(0,1)|\sup_{x\in\mathcal S}f(x)
=:C_fh^q=O(h^q).
\label{eq:matching_design_cluster_variance}
\end{align}
Set $\xi_i=V_i-\E[V_i]$. The $\xi_i$ are independent across $i$, and \eqref{eq:matching_design_cluster_variance} gives
\begin{align}
\E[\xi_i]&=0,\qquad
|\xi_i|\le |V_i|+\E[|V_i|]\le2,\nonumber\\
\sum_{i=1}^n\E[\xi_i^2]
&=\sum_{i=1}^n\Var(V_i)\le nC_fh^q.
\label{eq:matching_design_bernstein_moments}
\end{align}
Using \eqref{eq:matching_local_design} and \eqref{eq:matching_design_cluster_entry}, Bernstein's inequality with \eqref{eq:matching_design_bernstein_moments} gives, for $0<t\le1$,
\begin{align}
\P\left(\left|[Q_{n,h}(s)-\E[Q_{n,h}(s)]]_{k,l}\right|>t\right)
&=\P\left(\left|\sum_{i=1}^n\xi_i\right|>nh^qt\right)\nonumber\\
&\le2\exp\left\{-\frac{n^2h^{2q}t^2}
{2\{\sum_{i=1}^n\E[\xi_i^2]+\frac23nh^qt\}}\right\}\nonumber\\
&\le2\exp\left\{-\frac{nh^qt^2}{2(C_f+2t/3)}\right\}\nonumber\\
&\le2\exp(-cnh^qt^2),
\label{eq:matching_design_bernstein}
\end{align}
where $c=\{2(C_f+2/3)\}^{-1}>0$ is independent of $s,k,l$.

Partition $\R^q$ into cubes of side $hn^{-3}$; choose one point of $\mathcal S$ in each occupied cube to form $\mathcal G_n$. By \eqref{eq:matching_bandwidth}, $\frac{h^{-q}}n=\frac1{nh^q}\to0$, so
\begin{equation}
\sup_{s\in\mathcal S}\min_{g\in\mathcal G_n}\norm{s-g}
\le\sqrt q\,hn^{-3},\qquad
|\mathcal G_n|
\le\left(2+\frac{\operatorname{diam}(\mathcal S)}{hn^{-3}}\right)^q
\lesssim h^{-q}n^{3q}\lesssim n^{3q+1}.
\label{eq:matching_grid_cover}
\end{equation}
Set $t_n=\sqrt{(3q+2)\log n/(cnh^q)}\to0$. For large $n$, \eqref{eq:matching_design_bernstein} and \eqref{eq:matching_grid_cover} give
\begin{align}
&\P\left(
\max_{\substack{s\in\mathcal G_n\\1\le k,l\le q+1}}
\left|[Q_{n,h}(s)-\E[Q_{n,h}(s)]]_{k,l}\right|>t_n
\right)\nonumber\\
&\qquad\le\sum_{s\in\mathcal G_n}\sum_{k,l=1}^{q+1}
\P\left(\left|[Q_{n,h}(s)-\E[Q_{n,h}(s)]]_{k,l}\right|>t_n\right)\nonumber\\
&\qquad\le2(q+1)^2|\mathcal G_n|\exp(-cnh^qt_n^2)\nonumber\\
&\qquad\lesssim n^{3q+1}\exp\{-(3q+2)\log n\}
=n^{-1}\to0.
\label{eq:matching_grid_design_tail}
\end{align}
By \eqref{eq:matching_grid_design_tail},
\begin{align}
&\max_{s\in\mathcal G_n}\opnorm{Q_{n,h}(s)-\E[Q_{n,h}(s)]}\nonumber\\
&\qquad\le(q+1)
\max_{\substack{s\in\mathcal G_n\\1\le k,l\le q+1}}
\left|[Q_{n,h}(s)-\E[Q_{n,h}(s)]]_{k,l}\right|
=O_\P\left(\sqrt{\frac{\log n}{nh^q}}\right).
\label{eq:matching_grid_design_rate}
\end{align}

Set $\Delta_n=\sqrt q\,hn^{-3}$. Define $\mathcal H_n(g)=\{x\in\R^q:|\norm{x-g}-h|\le\Delta_n\}$.
For large $n$, $\Delta_n\le h/2$, so
\begin{align}
\P(S_{i,j}\in\mathcal H_n(g))
&\le\left(\sup_{\mathcal S}f\right)|B(0,1)|
\{(h+\Delta_n)^q-(h-\Delta_n)^q\}\nonumber\\
&\le2q\left(\sup_{\mathcal S}f\right)|B(0,1)|
(h+\Delta_n)^{q-1}\Delta_n
\le C_Hh^qn^{-3},
\label{eq:matching_shell_probability}
\end{align}
for a constant $C_H$. Set $H_{i,n}(g)=J^{-1}\sum_{j=1}^J\mathbf1\{S_{i,j}\in\mathcal H_n(g)\}\in[0,1]$. By \eqref{eq:matching_shell_probability},
\begin{align}
\E[H_{i,n}(g)]
&=\frac1J\sum_{j=1}^J\P(S_{i,j}\in\mathcal H_n(g))
\le C_Hh^qn^{-3},\nonumber\\
\Var(H_{i,n}(g))
&\le\E[H_{i,n}(g)^2]\le\E[H_{i,n}(g)]\le C_Hh^qn^{-3}.
\label{eq:matching_shell_cluster_moments}
\end{align}
Independence across $i$, \eqref{eq:matching_shell_cluster_moments}, Bernstein's inequality, and \eqref{eq:matching_grid_cover} give, for $x>0$,
\begin{equation}\label{eq:matching_shell_tail}
\begin{aligned}
&\P\Bigg(
\max_{g\in\mathcal G_n}
\left|\frac1n\sum_{i=1}^n\{H_{i,n}(g)-\E[H_{i,n}(g)]\}\right|>
\sqrt{\frac{2C_Hh^qn^{-3}x}{n}}+\frac{2x}{3n}
\Bigg)\\
&\qquad\le\sum_{g\in\mathcal G_n}2e^{-x}
=2|\mathcal G_n|e^{-x}\lesssim n^{3q+1}e^{-x}.
\end{aligned}
\end{equation}
Taking $x=(3q+2)\log n$ in \eqref{eq:matching_shell_tail} and using \eqref{eq:matching_shell_cluster_moments} yields
\begin{equation}\label{eq:matching_design_shell}
\max_{g\in\mathcal G_n}\frac1{nh^q}\sum_{i=1}^n H_{i,n}(g)
=O_\P\left(
n^{-3}+\sqrt{\frac{n^{-3}\log n}{nh^q}}
+\frac{\log n}{nh^q}
\right).
\end{equation}
For $\norm{u},\norm{\tilde u}\le1$,
\begin{align}
z_k(u)z_l(u)-z_k(\tilde u)z_l(\tilde u)
&=\{z_k(u)-z_k(\tilde u)\}z_l(u)+z_k(\tilde u)\{z_l(u)-z_l(\tilde u)\},\nonumber\\
\Longrightarrow\quad |z_k(u)z_l(u)-z_k(\tilde u)z_l(\tilde u)|
&\le |z_k(u)-z_k(\tilde u)|+|z_l(u)-z_l(\tilde u)|
\le2\norm{u-\tilde u}.
\label{eq:matching_polynomial_increment}
\end{align}
For each $s\in\mathcal S$, choose $g\in\mathcal G_n$ with $\norm{s-g}\le\Delta_n$, using \eqref{eq:matching_grid_cover}. Then
\begin{equation}\label{eq:matching_argument_distances}
\norm{u_{i,j}(s)-u_{i,j}(g)}
=\frac{\norm{s-g}}h\le\frac{\Delta_n}h,\qquad
\big|\norm{S_{i,j}-s}-\norm{S_{i,j}-g}\big|
\le\norm{s-g}\le\Delta_n.
\end{equation}
If $K(u_{i,j}(s))\ne K(u_{i,j}(g))$, the definition of $K$ and \eqref{eq:matching_argument_distances} give
\begin{align}
&\min\{\norm{S_{i,j}-s},\norm{S_{i,j}-g}\}
\le h<\max\{\norm{S_{i,j}-s},\norm{S_{i,j}-g}\},\nonumber\\
&\big|\norm{S_{i,j}-g}-h\big|
\le\big|\norm{S_{i,j}-g}-\norm{S_{i,j}-s}\big|\le\Delta_n.
\label{eq:matching_kernel_change_shell}
\end{align}
By \eqref{eq:matching_kernel_change_shell}, $S_{i,j}\in\mathcal H_n(g)$ whenever $K(u_{i,j}(s))\neq K(u_{i,j}(g))$. Since $K(u)\in\{0,1\}$,
\begin{equation}\label{eq:matching_kernel_increment}
|K(u_{i,j}(s))-K(u_{i,j}(g))|
\le\mathbf1\{S_{i,j}\in\mathcal H_n(g)\}.
\end{equation}
To derive \eqref{eq:matching_weighted_polynomial_increment}, for $u,\tilde u\in\R^q$, decompose
\begin{align}
&K(u)z_k(u)z_l(u)-K(\tilde u)z_k(\tilde u)z_l(\tilde u)\nonumber\\
&\quad=K(u)K(\tilde u)\{z_k(u)z_l(u)-z_k(\tilde u)z_l(\tilde u)\}\nonumber\\
&\qquad+K(u)\{1-K(\tilde u)\}z_k(u)z_l(u)
-\{1-K(u)\}K(\tilde u)z_k(\tilde u)z_l(\tilde u).
\label{eq:matching_kernel_product_decomposition}
\end{align}
Since $K(u)\in\{0,1\}$ and $K(u)|z_k(u)z_l(u)|\le K(u)$,
\begin{align}
&K(u)\{1-K(\tilde u)\}|z_k(u)z_l(u)|
+\{1-K(u)\}K(\tilde u)|z_k(\tilde u)z_l(\tilde u)|\nonumber\\
&\quad\le K(u)\{1-K(\tilde u)\}+\{1-K(u)\}K(\tilde u)
=|K(u)-K(\tilde u)|.
\label{eq:matching_kernel_difference_bound}
\end{align}
\noeqref{eq:matching_kernel_increment,eq:matching_kernel_product_decomposition}%
Moreover, $K(u)K(\tilde u)=1$ implies $\norm{u},\norm{\tilde u}\le1$, so \eqref{eq:matching_polynomial_increment} applies. The triangle inequality and \eqref{eq:matching_argument_distances}--\eqref{eq:matching_kernel_difference_bound} give
\begin{align}
&\big|K(u_{i,j}(s))z_k(u_{i,j}(s))z_l(u_{i,j}(s))
-K(u_{i,j}(g))z_k(u_{i,j}(g))z_l(u_{i,j}(g))\big|\nonumber\\
&\enspace\le K\!(u_{i,j}\!(s))K\!(u_{i,j}\!(g))
\big|z_k\!(u_{i,j}\!(s))z_l\!(u_{i,j}\!(s))-z_k\!(u_{i,j}\!(g))z_l\!(u_{i,j}\!(g))\big|
+|K\!(u_{i,j}\!(s))-K\!(u_{i,j}\!(g))|\nonumber\\
&\enspace\le2K(u_{i,j}(s))K(u_{i,j}(g))
\norm{u_{i,j}(s)-u_{i,j}(g)}
+|K(u_{i,j}(s))-K(u_{i,j}(g))|\nonumber\\
&\enspace\le\frac{2\Delta_n}h+\mathbf1\{S_{i,j}\in\mathcal H_n(g)\}.
\label{eq:matching_weighted_polynomial_increment}
\end{align}
By \eqref{eq:matching_local_design} and \eqref{eq:matching_weighted_polynomial_increment},
\begin{equation}\label{eq:matching_design_grid_increment}
\opnorm{Q_{n,h}(s)-Q_{n,h}(g)}
\le(q+1)\left\{\frac{2\Delta_n}{h^{q+1}}
+\frac1{nh^q}\sum_{i=1}^n H_{i,n}(g)\right\}.
\end{equation}
Taking expectations in \eqref{eq:matching_design_grid_increment} and using \eqref{eq:matching_shell_cluster_moments} gives
\begin{align}
\opnorm{\E[Q_{n,h}(s)]-\E[Q_{n,h}(g)]}
&\le\E[\opnorm{Q_{n,h}(s)-Q_{n,h}(g)}]\nonumber\\
&\le(q+1)\left\{\frac{2\Delta_n}{h^{q+1}}+C_Hn^{-3}\right\}.
\label{eq:matching_expected_design_grid_increment}
\end{align}
By \eqref{eq:matching_bandwidth},
\begin{equation}\label{eq:matching_grid_increment_rate}
\frac{\Delta_n}{h^{q+1}}
=\sqrt q\,n^{-3}h^{-q}\lesssim n^{-2}.
\end{equation}
Equations \eqref{eq:matching_grid_design_rate}, \eqref{eq:matching_design_shell}, \eqref{eq:matching_design_grid_increment}, \eqref{eq:matching_expected_design_grid_increment}, and \eqref{eq:matching_grid_increment_rate}, with $t_n\to0$ and $n^{-2}/t_n\lesssim n^{-3/2}/\sqrt{\log n}\to0$, give
\begin{align}
&\sup_{s\in\mathcal S}\opnorm{Q_{n,h}(s)-\E[Q_{n,h}(s)]}\nonumber\\
&\quad\le\max_{g\in\mathcal G_n}\opnorm{Q_{n,h}(g)-\E[Q_{n,h}(g)]}
+Cn^{-2}+\frac{q+1}{nh^q}\max_{g\in\mathcal G_n}\sum_{i=1}^n H_{i,n}(g)\nonumber\\
&\quad=O_\P(t_n+n^{-2}+n^{-3}+n^{-3/2}t_n+t_n^2)
=O_\P\left(\sqrt{\frac{\log n}{nh^q}}\right)=o_\P(1).
\label{eq:matching_uniform_design}
\end{align}
By \eqref{eq:matching_uniform_design}, the event
\begin{equation}\label{eq:matching_design_event}
\mathcal A_n=
\left\{\sup_{s\in\mathcal S}\opnorm{Q_{n,h}(s)-\E[Q_{n,h}(s)]}
\le\frac{c_Q}{2}\right\}
\end{equation}
has probability tending to one. On $\mathcal A_n$, \eqref{eq:matching_population_design_bounds}, \eqref{eq:matching_population_design_upper_bound}, and \eqref{eq:matching_design_event} give
\begin{align}
\inf_{s\in\mathcal S}\lambda_{\min}(Q_{n,h}(s))
&\ge\inf_{s\in\mathcal S}\lambda_{\min}(\E[Q_{n,h}(s)])
-\sup_{s\in\mathcal S}\opnorm{Q_{n,h}(s)-\E[Q_{n,h}(s)]}\nonumber\\
&\ge c_Q-\frac{c_Q}{2}=\frac{c_Q}{2}.\nonumber\\
\intertext{Thus,}
\sup_{s\in\mathcal S}\opnorm{Q_{n,h}(s)^{-1}}
&=\left\{\inf_{s\in\mathcal S}\lambda_{\min}(Q_{n,h}(s))\right\}^{-1}
\le\frac2{c_Q}.\nonumber\\
\intertext{In addition,}
\sup_{s\in\mathcal S}\opnorm{Q_{n,h}(s)}
&\le\sup_{s\in\mathcal S}\opnorm{\E[Q_{n,h}(s)]}
+\sup_{s\in\mathcal S}\opnorm{Q_{n,h}(s)-\E[Q_{n,h}(s)]}\nonumber\\
&\le\sup_{s\in\mathcal S}\opnorm{\E[Q_{n,h}(s)]}+\frac{c_Q}{2}\lesssim1.
\label{eq:matching_empirical_design_bounds}
\end{align}

To separate the smoothing and sampling errors, let
\begin{equation}\label{eq:matching_row_residual}
\epsilon_{i,j}=Y_{i,j}-\mu(S_{i,j}),
\qquad
d_{i,j}(s)=\mu(S_{i,j})-\mu(s)-\nabla\mu(s)'(S_{i,j}-s).
\end{equation}
Taylor's theorem on the convex support gives $|d_{i,j}(s)|\lesssim h^2$ whenever $K(u_{i,j}(s))=1$. Define
\begin{align}
T_{n,h}(s)
&=\frac{1}{nJh^q}\sum_{i,j}
K(u_{i,j}(s))z(u_{i,j}(s))d_{i,j}(s),\nonumber\\
U_{n,h}(s)
&=\frac{1}{nJh^q}\sum_{i,j}
K(u_{i,j}(s))z(u_{i,j}(s))\epsilon_{i,j}.
\label{eq:matching_score_decomposition}
\end{align}
Hence, on $\mathcal A_n$,
\begin{equation}\label{eq:matching_fit_error_decomposition}
\hat\mu_h(s)-\mu(s)
=e_1'Q_{n,h}(s)^{-1}T_{n,h}(s)
+e_1'Q_{n,h}(s)^{-1}U_{n,h}(s),
\qquad
\sup_{s\in\mathcal S}\norm{T_{n,h}(s)}\lesssim h^2.
\end{equation}
For the bound on $T_{n,h}$, use $\norm{z(u)}\le\sqrt2$ and $\left|[Q_{n,h}]_{1,1}\right|\le\opnorm{Q_{n,h}}\lesssim1$ from \eqref{eq:matching_empirical_design_bounds}.

By Assumption \ref{ass:local-linear-regularity}, we have
\begin{align}
\E\left[\frac1J\sum_{j=1}^J
K(u_{i,j}(s))z(u_{i,j}(s))\epsilon_{i,j}\right]
&=\E\left[
K\left(\frac{S-s}{h}\right)z\left(\frac{S-s}{h}\right)
\{Y-\mu(S)\}\right]=0.
\label{eq:matching_cluster_score_centering}
\end{align}
Therefore,
\begin{align}
\E[\norm{U_{n,h}(s)}^2]
&=\frac{1}{n^2h^{2q}}\sum_{i=1}^n
\E\left[\norm{\frac1J\sum_{j=1}^J
K(u_{i,j}(s))z(u_{i,j}(s))\epsilon_{i,j}}^2\right]\nonumber\\
&\le\frac{1}{n^2Jh^{2q}}\sum_{i=1}^n\sum_{j=1}^J
\E\left[K(u_{i,j}(s))^2\norm{z(u_{i,j}(s))}^2\epsilon_{i,j}^2\right]\nonumber\\
&=\frac{1}{nh^{2q}}\,
\E\left[
K\left(\frac{S-s}{h}\right)^2
\norm{z\left(\frac{S-s}{h}\right)}^2
\{Y-\mu(S)\}^2\right]\nonumber\\
&=\frac{1}{nh^{2q}}\int_{\mathcal S\cap B(s,h)}
\norm{z\left(\frac{x-s}{h}\right)}^2\sigma^2(x)f(x)\,dx\nonumber\\
&\le\frac{2|B(0,1)|h^q}{nh^{2q}}
\left(\sup_{x\in\mathcal S}\sigma^2(x)\right)
\left(\sup_{x\in\mathcal S}f(x)\right)
\lesssim\frac{1}{nh^q}.
\label{eq:matching_score_second_moment}
\end{align}
Using the same argument as in \eqref{eq:matching_score_second_moment}, we obtain
\begin{equation}\label{eq:matching_design_second_moment}
\sup_{s\in\mathcal S}
\E[\opnorm{Q_{n,h}(s)-\E[Q_{n,h}(s)]}^2]
\lesssim\frac{1}{nh^q}.
\end{equation}
Integrating \eqref{eq:matching_fit_error_decomposition} and using \eqref{eq:matching_score_second_moment},
\begin{equation}\label{eq:matching_integrated_error_bound}
\mathbf1_{\mathcal A_n}
\int\{\hat\mu_h(s)-\mu(s)\}^2\,d\nu(s)
\lesssim h^4+\int\norm{U_{n,h}(s)}^2\,d\nu(s),\quad
\E\left[\int\norm{U_{n,h}(s)}^2\,d\nu(s)\right]
\lesssim\frac{1}{nh^q}.
\end{equation}
Markov's inequality and $\P(\mathcal A_n)\to1$ prove \eqref{eq:matching_local_L2}.

To prove \eqref{eq:matching_local_mean}, we first recall the identity
\begin{equation}\label{eq:matching_inverse_identity}
Q_{n,h}^{-1}-\E[Q_{n,h}]^{-1}
=Q_{n,h}^{-1}(\E[Q_{n,h}]-Q_{n,h})\E[Q_{n,h}]^{-1},
\end{equation}
which separates the deterministic inverse from its estimation error. On $\mathcal A_n$, \eqref{eq:matching_inverse_identity} implies
\begin{align}
\left|\int e_1'(Q_{n,h}^{-1}-\E[Q_{n,h}]^{-1})U_{n,h}\,d\nu\right|
&\lesssim
\int\opnorm{Q_{n,h}-\E[Q_{n,h}]}\norm{U_{n,h}}\,d\nu,\label{eq:matching_integrated_inverse_error}\\
\E\left[\int\opnorm{Q_{n,h}-\E[Q_{n,h}]}\norm{U_{n,h}}\,d\nu\right]
&\le\int
\sqrt{\E[\opnorm{Q_{n,h}-\E[Q_{n,h}]}^2]\E[\norm{U_{n,h}}^2]}\,d\nu
\lesssim\frac{1}{nh^q}.\nonumber
\end{align}
Thus \eqref{eq:matching_integrated_inverse_error} is $O_\P((nh^q)^{-1})=o_\P(h^2)$ by \eqref{eq:matching_bandwidth}.

Integrating \eqref{eq:matching_fit_error_decomposition}, on $\mathcal A_n$,
\begin{align}
\int\{\hat\mu_h(s)-\mu(s)\}\,d\nu(s)
&=\int e_1'Q_{n,h}(s)^{-1}T_{n,h}(s)\,d\nu(s)\nonumber\\
&\quad+\int e_1'\{Q_{n,h}(s)^{-1}-\E[Q_{n,h}(s)]^{-1}\}U_{n,h}(s)\,d\nu(s)\nonumber\\
&\quad+\int e_1'\E[Q_{n,h}(s)]^{-1}U_{n,h}(s)\,d\nu(s).
\label{eq:matching_integrated_mean_decomposition}
\end{align}
By \eqref{eq:matching_empirical_design_bounds} and \eqref{eq:matching_fit_error_decomposition}, on $\mathcal A_n$,
\begin{align}
\left|\int e_1'Q_{n,h}(s)^{-1}T_{n,h}(s)\,d\nu(s)\right|
&\le\sup_{s\in\mathcal S}\opnorm{Q_{n,h}(s)^{-1}}
\sup_{s\in\mathcal S}\norm{T_{n,h}(s)}
\lesssim h^2.
\label{eq:matching_integrated_smoothing_bound}
\end{align}
Note that
\begin{align}
&\int e_1'\E[Q_{n,h}(s)]^{-1}U_{n,h}(s)\,d\nu(s)
=\frac1n\sum_{i=1}^n\psi_{i,h},\nonumber\\
&\psi_{i,h}
=\frac{1}{Jh^q}\sum_{j=1}^J\epsilon_{i,j}
\int K(u_{i,j}(s))e_1'\E[Q_{n,h}(s)]^{-1}z(u_{i,j}(s))\,d\nu(s).
\label{eq:matching_integrated_cluster_score}
\end{align}
Equation \eqref{eq:matching_cluster_score_centering} gives $\E[\psi_{i,h}]=0$. Since $\nu$ has bounded density and the integrand is supported on $B(S_{i,j},h)$, \eqref{eq:matching_population_design_bounds} gives
\begin{equation}\label{eq:matching_integrated_cluster_variance}
|\psi_{i,h}|\lesssim\frac1J\sum_j|\epsilon_{i,j}|,
\qquad
\E[\psi_{i,h}^2]\lesssim\frac1J\sum_j\E[\epsilon_{i,j}^2]\lesssim1.
\end{equation}
By Chebyshev's inequality, we have
\begin{equation}\label{eq:matching_cluster_average_rate}
\frac1n\sum_{i=1}^n\psi_{i,h}=O_\P(n^{-1/2}).
\end{equation}
\noeqref{eq:matching_integrated_mean_decomposition,eq:matching_integrated_smoothing_bound,eq:matching_integrated_cluster_score}%
Combining \eqref{eq:matching_integrated_inverse_error}--\eqref{eq:matching_cluster_average_rate} proves \eqref{eq:matching_local_mean}.
\end{proof}

\begin{proof}[Proof of Theorem \ref{thm:decomp_matching}]
Let $\E^{(0)}$ denote the expectation conditional on the reservoir $\mathcal R$. Note that Assumption \ref{ass:pop-dist} implies the target experiment is independent of $\mathcal R$.
For each method $M$, write
\begin{equation}\label{eq:matching_estimation_error}
e_M=\hat\tau_M-g_M(\D^{(0)}).
\end{equation}
The exact decomposition \eqref{eq:MSE_step3}, conditional on $\mathcal R$, is
\begin{equation}\label{eq:matching_exact_RE}
\mathrm{RE}(\hat\tau_M)
=\Var^{(0)}(e_M)
+2\mathrm{Cov}^{(0)}\left(e_M,g_M(\D^{(0)})-\E[\tau_T^{(0)}\mid\D^{(0)}]\right).
\end{equation}
We prove the results for each method separately.

\bigskip
\noindent\textbf{Proof for ArmS-m. }Let $f_w$ and $\mathcal S_w$ denote the density and support of $m_{\sh,w}$. Define
\begin{equation}\label{eq:matching_regressions}
r_w(m)=\E[m_{T,w}\mid m_{\sh,w}=m],\qquad w\in\{0,1\}.
\end{equation}

By Assumption \ref{ass:local-linear-regularity}, apply Lemma \ref{lem:matching_local_linear_rates} with $J=1$, $S=m_{\sh,w}$, $Y=m_{T,w}$, and $\nu$ the distribution of $m_{\sh,w}^{(0)}$, for each $w\in\{0,1\}$. Note that Assumption \ref{ass:matching-bandwidth} holds for $q=t_0$. Then
\begin{align}
g_{\ArmSm}(\D)&=r_1(m_{\sh,1})-r_0(m_{\sh,0}),\nonumber\\
e_{\ArmSm}
&=\{\hat r_1(m_{\sh,1}^{(0)})-r_1(m_{\sh,1}^{(0)})\}
-\{\hat r_0(m_{\sh,0}^{(0)})-r_0(m_{\sh,0}^{(0)})\}.
\label{eq:matching_arm_error_decomposition}
\end{align}
By Lemma \ref{lem:matching_local_linear_rates} and \eqref{eq:matching_arm_error_decomposition},
\begin{equation}\label{eq:matching_difference_bounds}
\begin{aligned}
\E^{(0)}[e_{\ArmSm}^2]
&\le2\sum_{w=0}^1\int_{\mathcal S_w}
\{\hat r_w(m)-r_w(m)\}^2f_w(m)\,dm\\
&=O_\P\left(h^4+\frac1{nh^{t_0}}\right),\\
|\E^{(0)}[e_{\ArmSm}]|
&=\left|\sum_{w=0}^1(2w-1)\int_{\mathcal S_w}
\{\hat r_w(m)-r_w(m)\}f_w(m)\,dm\right|\\
&\le\sum_{w=0}^1\left|\int_{\mathcal S_w}
\{\hat r_w(m)-r_w(m)\}f_w(m)\,dm\right|\\
&=O_\P\left(h^2+\frac1{\sqrt n}\right).
\end{aligned}
\end{equation}

For $M=\ArmSm$, \eqref{eq:matching_difference_bounds}, Assumption \ref{ass:local-linear-regularity}, and Cauchy--Schwarz in \eqref{eq:matching_exact_RE} give
\begin{align}
|\mathrm{RE}(\hat\tau_M)|
&\le\Var^{(0)}(e_M)
+2\sqrt{\Var^{(0)}(e_M)\mathrm{AE}(\hat\tau_M)}\nonumber\\
&=O_\P\left(
h^4+\frac1{nh^{t_0}}+
\sqrt{h^4+\frac1{nh^{t_0}}}
\right)\nonumber\\
&=O_\P\left(h^2+\frac1{\sqrt{nh^{t_0}}}\right).
\label{eq:matching_covariance_rate}
\end{align}
This proves \eqref{eq:matching_RE_arm_pooled} for ArmS-m.

The law of iterated expectations gives
\begin{align}
\E[g_{\ArmSm}(\D^{(0)})]
&=\E[r_1(m_{\sh,1}^{(0)})]-\E[r_0(m_{\sh,0}^{(0)})]\nonumber\\
&=\E[m_{T,1}]-\E[m_{T,0}]=\E[\tau_T].
\label{eq:matching_arm_mean_target}
\end{align}
Equations \eqref{eq:matching_arm_mean_target} and \eqref{eq:matching_difference_bounds} give, for $M=\ArmSm$,
\begin{equation}\label{eq:matching_bias_from_mean_error}
\mathrm{Bias}^2(\hat\tau_M)
=\{\E^{(0)}[e_M]\}^2
=O_\P\left((h^2+n^{-1/2})^2\right)
=O_\P(h^4+n^{-1}).
\end{equation}
This proves \eqref{eq:matching_bias_ERS_arm} for ArmS-m.

Equation \eqref{eq:matching_arm_error_decomposition} and Assumption \ref{ass:local-linear-regularity} give
\begin{equation}\label{eq:matching_arm_AE}
\mathrm{AE}(\hat\tau_{\ArmSm})
=\Var\left(r_1(m_{\sh,1})-r_0(m_{\sh,0})
-\E[\tau_T\mid\D]\right)=O(1).
\end{equation}
This proves \eqref{eq:AE_matching} for ArmS-m. In general, $r_1(m_{\sh,1})-r_0(m_{\sh,0})-\E[\tau_T\mid\D]$ need not vanish. To see this, note that $m_{T,1}=m_{T,0}+\tau_T$, and Assumption \ref{ass:cond-indep-tau-alpha} gives
\begin{align}
&r_1(m_{\sh,1})-r_0(m_{\sh,0})-\E[\tau_T\mid\D]\nonumber\\
&=\E[m_{T,0}+\tau_T\mid m_{\sh,1}]
-\E[m_{T,0}\mid m_{\sh,0}]-\E[\tau_T\mid\tau_{\sh}]\nonumber\\
&=\{\E[m_{T,0}\mid m_{\sh,1}]-\E[m_{T,0}\mid m_{\sh,0}]\}
+\{\E[\tau_T\mid m_{\sh,1}]-\E[\tau_T\mid\tau_{\sh}]\}.
\label{eq:matching_arm_residual_expansion}
\end{align}
Assumption \ref{ass:cond-indep-tau-alpha} does not require the two differences in \eqref{eq:matching_arm_residual_expansion} to cancel.

\bigskip
\noindent\textbf{Proof for ERS. }By Assumption \ref{ass:local-linear-regularity}, apply Lemma \ref{lem:matching_local_linear_rates} with $J=1$, $S=\tau_{\sh}$, $Y=\tau_T$, $q=t_0$, and $\nu$ the distribution of $\tau_{\sh}^{(0)}$. Define
\begin{equation}\label{eq:matching_ERS_target_error}
r_E(\tau)=\E[\tau_T\mid\tau_{\sh}=\tau],
\qquad
g_{\ERS}(\D)=r_E(\tau_{\sh}).
\end{equation}
Equations \eqref{eq:matching_local_L2} and \eqref{eq:matching_local_mean} give
\begin{equation}\label{eq:matching_ERS_estimation_rates}
\E^{(0)}[e_{\ERS}^2]
=O_\P\left(h^4+\frac{1}{nh^{t_0}}\right),
\qquad
\E^{(0)}[e_{\ERS}]
=O_\P\left(h^2+\frac1{\sqrt n}\right).
\end{equation}

\noindent By definition of $e_{\ERS}$ and $r_E$,
\begin{equation}\label{eq:matching_ERS_covariance}
\mathrm{Cov}^{(0)}\left(e_{\ERS},r_E(\tau_{\sh}^{(0)})-\E[\tau_T^{(0)}\mid\D^{(0)}]\right)=0.
\end{equation}
Equations \eqref{eq:matching_exact_RE} and \eqref{eq:matching_ERS_estimation_rates} prove \eqref{eq:matching_RE_ERS}.

The law of iterated expectations gives
\begin{equation}\label{eq:matching_mean_targets}
\E[g_{\ERS}(\D^{(0)})]
=\E[r_E(\tau_{\sh}^{(0)})]=\E[\tau_T].
\end{equation}
Equations \eqref{eq:matching_mean_targets} and \eqref{eq:matching_ERS_estimation_rates} give \eqref{eq:matching_bias_from_mean_error} for $M=\ERS$, proving \eqref{eq:matching_bias_ERS_arm} for ERS.

By Assumption \ref{ass:cond-indep-tau-alpha},
\begin{equation}\label{eq:matching_conditional_mean_reduction}
\E[\tau_T\mid\D]
=\E[\tau_T\mid\tau_{\sh},m_{\sh,0}]
=\E[\tau_T\mid\tau_{\sh}]
=r_E(\tau_{\sh}).
\end{equation}
Equations \eqref{eq:matching_ERS_target_error} and \eqref{eq:matching_conditional_mean_reduction} give $\mathrm{AE}(\hat\tau_{\ERS})=0$, proving \eqref{eq:AE_matching} for ERS.

\bigskip
\noindent\textbf{Proof for AggS-m. }Define
\begin{equation}\label{eq:matching_pooled_regression}
f_G(m)=\frac{f_0(m)+f_1(m)}{2},
\qquad
r_G(m)=\frac{f_0(m)r_0(m)+f_1(m)r_1(m)}{f_0(m)+f_1(m)}.
\end{equation}

In Lemma \ref{lem:matching_local_linear_rates}, take $J=2$ and let $((S_{i,1},Y_{i,1}),(S_{i,2},Y_{i,2}))$ be a uniform random permutation of $((m_{\sh,0}^{(i)},m_{T,0}^{(i)}),(m_{\sh,1}^{(i)},m_{T,1}^{(i)}))$, with the permutations independent across $i$ and of $\mathcal R$. The sum over $j$ in \eqref{eq:matching_local_fit} is unchanged by this permutation, so $\hat r_G=\hat\mu_h$. By \eqref{eq:matching_pooled_regression}, $f_G$ is the density of $S_{i,j}$, and
\begin{equation}\label{eq:matching_pooled_row_law}
\E[Y_{i,j}\mid S_{i,j}=m]=r_G(m),
\quad j=1,2.
\end{equation}
By Assumption \ref{ass:local-linear-regularity}, for $w\in\{0,1\}$,
\begin{equation}\label{eq:matching_pooled_evaluation_density}
0\le f_w(m)\le2f_G(m)\le2\sup_x f_G(x)<\infty,
\qquad
\mathcal S_w\subseteq\operatorname{supp}(f_G).
\end{equation}
Using \eqref{eq:matching_pooled_row_law} and \eqref{eq:matching_pooled_evaluation_density}, Lemma \ref{lem:matching_local_linear_rates} with $q=t_0$ and $d\nu(m)=f_w(m)\,dm$ gives
\begin{align}
\int_{\mathcal S_w}\{\hat r_G(m)-r_G(m)\}^2f_w(m)\,dm
&=O_\P\left(h^4+\frac{1}{nh^{t_0}}\right),\nonumber\\
\int_{\mathcal S_w}\{\hat r_G(m)-r_G(m)\}f_w(m)\,dm
&=O_\P\left(h^2+\frac1{\sqrt n}\right).
\label{eq:matching_pooled_evaluation_rates}
\end{align}
Define
\begin{align}
g_{\AggSm}(\D)&=r_G(m_{\sh,1})-r_G(m_{\sh,0}),\nonumber\\
e_{\AggSm}
&=\{\hat r_G(m_{\sh,1}^{(0)})-r_G(m_{\sh,1}^{(0)})\}
-\{\hat r_G(m_{\sh,0}^{(0)})-r_G(m_{\sh,0}^{(0)})\}.
\label{eq:matching_pooled_error_decomposition}
\end{align}
Using \eqref{eq:matching_pooled_error_decomposition} and \eqref{eq:matching_pooled_evaluation_rates} and using the same argument as \eqref{eq:matching_difference_bounds}, we obtain that
\begin{equation}\label{eq:matching_pooled_estimation_rates}
\E^{(0)}[e_{\AggSm}^2]
=O_\P\left(h^4+\frac1{nh^{t_0}}\right),
\qquad
\E^{(0)}[e_{\AggSm}]
=O_\P\left(h^2+\frac1{\sqrt n}\right).
\end{equation}
Similar to \eqref{eq:matching_covariance_rate}, we can prove \eqref{eq:matching_RE_arm_pooled} for AggS-m.

Substituting \eqref{eq:matching_pooled_regression} gives
\begin{align}
\E[g_{\AggSm}(\D^{(0)})]-\E[\tau_T]
&=\int\{(f_1-f_0)r_G-f_1r_1+f_0r_0\}\,dm\nonumber\\
&=\int
\frac{(f_1-f_0)(f_1r_1+f_0r_0)
-(f_1+f_0)(f_1r_1-f_0r_0)}
{f_1+f_0}\,dm\nonumber\\
&=\int_{\mathcal S_0\cap\mathcal S_1}
\frac{2f_0f_1}{f_0+f_1}(r_0-r_1)\,dm
=:b_{0,G}.
\label{eq:matching_pooled_bias_calculation}
\end{align}
Then,
\begin{align}
\mathrm{Bias}^2(\hat\tau_{\AggSm})
&=\{b_{0,G}+\E^{(0)}[e_{\AggSm}]\}^2\nonumber\\
&=b_{0,G}^2
+2b_{0,G}\E^{(0)}[e_{\AggSm}]
+\{\E^{(0)}[e_{\AggSm}]\}^2.
\label{eq:matching_pooled_bias_square}
\end{align}
Substituting \eqref{eq:matching_pooled_estimation_rates} into \eqref{eq:matching_pooled_bias_square} proves both cases of \eqref{eq:matching_bias_pooled}.

By Equation \eqref{eq:matching_pooled_error_decomposition} and Assumption \ref{ass:local-linear-regularity},
\begin{equation}\label{eq:matching_pooled_AE}
\mathrm{AE}(\hat\tau_{\AggSm})
=\Var\left(r_G(m_{\sh,1})-r_G(m_{\sh,0})
-\E[\tau_T\mid\D]\right)=O(1),
\end{equation}
proving \eqref{eq:AE_matching} for AggS-m.
\end{proof}

\subsection{Proof of Theorem \ref{thm:decomp_ersplus}}\label{subapp:proof_thm_decomp_ersplus}

\begin{proof}
Apply Lemma \ref{lem:matching_local_linear_rates} with $J=1$, $S=\D$, $Y=\tau_T$, $q=2t_0$, $h=h_+$, and $\nu$ the distribution of $\D^{(0)}$. Clearly, the assumptions of Lemma \ref{lem:matching_local_linear_rates} hold under Assumptions \ref{ass:pop-dist}, \ref{ass:local-linear-regularity}, \ref{ass:matching-bandwidth}. Define
\begin{equation}\label{eq:ersplus_estimation_error}
e_{\ERSplus}=\hat\tau_{\ERSplus}-r_+(\D^{(0)}).
\end{equation}
Equations \eqref{eq:matching_local_L2} and \eqref{eq:matching_local_mean} give
\begin{equation}\label{eq:ersplus_estimation_rates}
\E^{(0)}[e_{\ERSplus}^2]
=O_\P\left(h_+^4+\frac{1}{nh_+^{2t_0}}\right),
\qquad
\E^{(0)}[e_{\ERSplus}]
=O_\P\left(h_+^2+\frac1{\sqrt n}\right).
\end{equation}
In particular, $\Var^{(0)}(e_{\ERSplus})\le\E^{(0)}[e_{\ERSplus}^2]=o_\P(1)$, as required by Proposition \ref{prop:MSE_decomp}. Note that $r_+(\D)=\E[\tau_T\mid\D]$. Thus,
\begin{equation}\label{eq:ersplus_target_identities}
r_+(\D)-\E[\tau_T\mid\D]
=0,\qquad
\E^{(0)}[r_+(\D^{(0)})]
=\E[r_+(\D)]
=\E[\E[\tau_T\mid\D]]
=\E[\tau_T].
\end{equation}
Thus,
\begin{align}
\mathrm{RE}(\hat\tau_{\ERSplus})
&=\Var^{(0)}(e_{\ERSplus})
\le\E^{(0)}[e_{\ERSplus}^2],\nonumber\\
\mathrm{Bias}^2(\hat\tau_{\ERSplus})
&=\{\E^{(0)}[e_{\ERSplus}]\}^2
=O_\P(h_+^4+n^{-1}),\nonumber\\
\mathrm{AE}(\hat\tau_{\ERSplus})
&=\Var\{r_+(\D)-\E[\tau_T\mid\D]\}=0.
\label{eq:ersplus_error_components}
\end{align}
Equations \eqref{eq:ersplus_estimation_rates} and \eqref{eq:ersplus_error_components} prove \eqref{eq:ersplus_RE}, \eqref{eq:ersplus_bias}, and \eqref{eq:ersplus_AE}.
\end{proof}

\section{Algorithms}\label{subapp:algorithms}

\subsection{The Unmatched ERS Algorithm (ERS-nm)}

See Algorithm \ref{alg:ers-unm}.

\begin{algorithm}
    \caption{Experiment Reservoir Surrogate without Matching (ERS-nm)}\label{alg:ers-unm}
    \begin{algorithmic}
        \STATE {\bfseries Input:} vectors of long-term and short-term ATEs from the reservoir $y=(\ATE_{T}^{(1)}, \ldots, \ATE_{T}^{(n)})'$ and\\ $X^{(i)}=(1, \ATE_{1}^{(i)}, \ldots, \ATE_{t_0}^{(i)})'$ for $i\in [n]$,  
        vector of short-term ATEs from the target experiment $X^{(0)}=(1, \ATE_{1}^{(0)}, \ldots, \ATE_{t_0}^{(0)})'$.

        \vspace{0.2cm}
        \STATE {\bfseries Procedure:}

        \STATE Fit a predictor $\hat{f}(x)$ on $(X, y)$.
        
        \vspace{0.2cm}
        \STATE {\bfseries Output:} $\hat y^{(0)}=\hat{f}(X^{(0)})$        
    \end{algorithmic}
    
\end{algorithm}

\subsection{The ERS with Control Covariates Algorithm (ERS+)}

See Algorithm \ref{alg:ers-with-control}.

\begin{algorithm}
    \caption{Experiment Reservoir Surrogate with Control Covariates (ERS+)}\label{alg:ers-with-control}
    \begin{algorithmic}
        \STATE {\bfseries Input:} long-term ATEs $y=(\ATE_T^{(1)},\ldots,\ATE_T^{(n)})'$ and joint short-term covariates $Z^{(i)}=(1,\ATE_1^{(i)},\ldots,\ATE_{t_0}^{(i)},\bar Y_{1,0}^{(i)},\ldots,\bar Y_{t_0,0}^{(i)})'$ for $i\in[n]$, the corresponding target vector $Z^{(0)}$, distance metric $\delta$, integer $k$.

        \vspace{0.2cm}
        \STATE {\bfseries Procedure:} 

        \STATE Set $Z=(Z^{(1)},\ldots,Z^{(n)})'$ and find matches $\tilde Z=M_\delta(Z,Z^{(0)};k)$.

        \STATE Find the corresponding long-term ATEs $\tilde y$ from $y$.
        
        \STATE Fit a predictor $\hat f(z)$ on $(\tilde Z,\tilde y)$.
        
        \vspace{0.2cm}
        \STATE {\bfseries Output:} $\hat\tau_{\ERSplus}=\hat f(Z^{(0)})$.
    \end{algorithmic}
    
\end{algorithm}

\subsection{The Matched Aggregate Surrogate Algorithm (AggS-m)}

See Algorithm \ref{alg:agg-surr-variant1}.

\begin{algorithm}[H]
\caption{Aggregate Surrogate with Matching (AggS-m)}\label{alg:agg-surr-variant1}
    \begin{algorithmic}
        \STATE {\bfseries Input:}  vectors of long-term and short-term group averages from the reservoir $y_w=(\bar Y_{T,w}^{(1)}, \ldots, \bar Y_{T,w}^{(n)})'$ and $X_{w}^{(i)} = (1, \bar Y_{1,w}^{(i)},\ldots, \bar Y_{t_0,w}^{(i)})'$ for $i\in [n]$ and $w\in\{0,1\}$, vectors of short-term group averages from the target experiments $X_{w}^{(0)}=(1, \bar Y_{1,w}^{(0)}, \ldots, \bar Y_{t_0, w}^{(0)})'$ for $w\in\{0,1\}$. distance metric $\delta$, integer $k$.

        \vspace{0.2cm}
        \STATE {\bfseries Procedure:}

        \STATE Set $X_w = (X_{w}^{(1)}, \ldots, X_{w}^{(n)})'$ for $w\in \{0, 1\}$.
        \STATE Pool the reservoir rows: $X=\begin{bmatrix}X_0\\X_1\end{bmatrix}$ and $y=\begin{bmatrix}y_0\\y_1\end{bmatrix}$.

        \STATE For each $w\in\{0,1\}$, find matches $\tilde X_w=M_\delta(X,X_w^{(0)};k)$ and the corresponding long-term group averages $\tilde y_w$ from $y$.
        
        \STATE Fit a predictor $\hat g_w(x)$ on $(\tilde X_w,\tilde y_w)$ for each $w\in\{0,1\}$.

        \vspace{0.2cm}
        \STATE {\bfseries Output:} $\hat{\tau}_T^{(0)} = \hat{g}_1(X_{1}^{(0)}) -\hat{g}_0(X_{0}^{(0)})$
    \end{algorithmic}
    
\end{algorithm}

\subsection{The Matched Armwise Surrogate Algorithm (ArmS-m)}

See Algorithm \ref{alg:arm-surr-variant1}.

\nopagebreak
\begin{algorithm}[H]
\caption{Aggregate Surrogate with Matching (ArmS-m)}\label{alg:arm-surr-variant1}
    \begin{algorithmic}
        \STATE {\bfseries Input:} Same as Algorithm \ref{alg:agg-surr-variant1}.

        \vspace{0.2cm}
        \STATE {\bfseries Procedure:}

        \STATE Set $X_w = (X_w^{(1)}, \ldots, X_w^{(n)})'$ for $w\in \{0, 1\}$.

        \STATE For each $w\in\{0,1\}$, find matches $\tilde X_w=M_\delta(X_w,X_w^{(0)};k)$

        \STATE Find the corresponding long-term group average $\tilde y_w$.

        \STATE Fit a predictor $\hat g_w(x)$ on $(\tilde X_w, \tilde y_w)$.

        \vspace{0.2cm}
        \STATE {\bfseries Output:} $\hat{\tau}_T^{(0)} = \hat{g}_1(X_1^{(0)}) -\hat{g}_0(X_0^{(0)})$
    \end{algorithmic}
    
\end{algorithm}

\section{Additional Figures}\label{subapp:figures}

Figure \ref{fig:ks} is a plot of the value of the number of matched metrics $k$ for each reservoir size and ERS method (with and without control covariates), shows as a percentage of the total number of metrics available in the experiment reservoir. Our matching process imposed a maximum of one third of the total, and in general the number of matched metrics ranges close to that maximum, increasing slightly on average as the reservoir size increases.

\begin{figure}[H]
    \centering
    \includegraphics[width=0.65\linewidth]{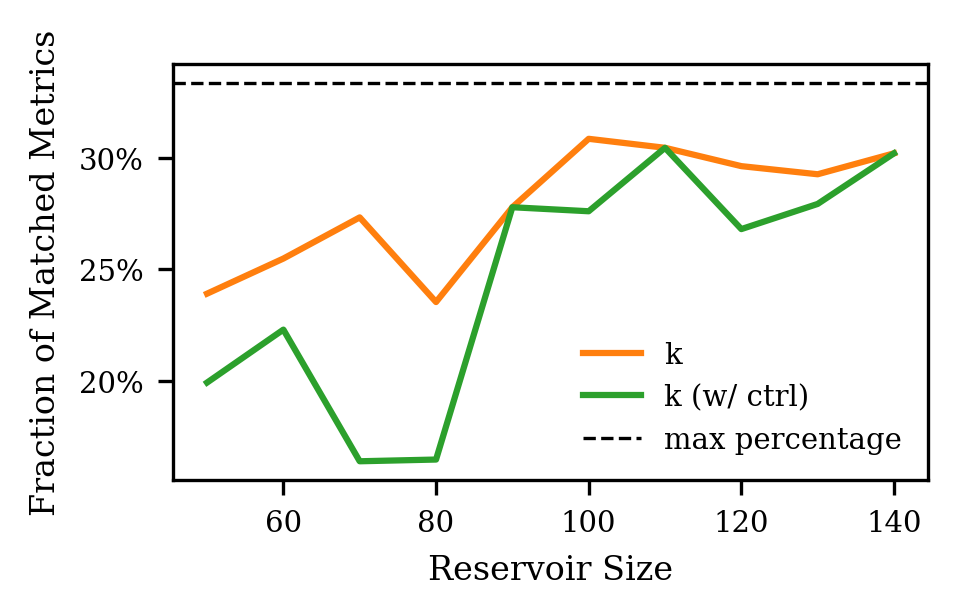}
    \caption{Percentage of reservoir metrics used for matching by ERS with and without control covariates. We give the comparison for the same sequence of reservoir sizes as in Figure~\ref{fig:percent-improvement}. The number of matched metrics $k$ is selected by leave-one-out cross-validation on the reservoir, with a maximum of one third of the total number of metrics, represented by the dashed line.}
    \label{fig:ks}
\end{figure}

\end{document}